\RequirePackage{fix-cm}
\documentclass[a4paper,twocolumn]{svjour3}
\pdfoutput=1
\smartqed
\usepackage{geometry}
\usepackage[pdftex]{graphicx}
\usepackage{wrapfig}
\usepackage[table]{xcolor}
\usepackage{amsmath}
\usepackage{amssymb}
\usepackage{array}
\usepackage{textcomp}
\usepackage{nicefrac}
\usepackage[T1]{fontenc}
\usepackage{booktabs}
\usepackage[utf8]{inputenc}	
\usepackage{bbm}
\usepackage{stmaryrd}
\usepackage{tikz}
\usetikzlibrary{arrows.meta,calc,through}
\usetikzlibrary{plotmarks}
\usetikzlibrary{matrix}
\usetikzlibrary{trees,positioning,fit}
\usetikzlibrary{intersections}
\usepackage{xifthen}
\usepackage{color}     
\usepackage{mathrsfs}
\usepackage{float}
\usepackage{makecell}
\usepackage{pgfplots}
\usepackage{thmtools}
\usepackage{nameref}
\usepackage{xspace}
\usepackage{authblk}
\usepackage{lipsum}

\usepackage{verbatim}

\pgfplotsset{compat=newest}

\usepackage{multirow}
\usepackage{multicol}

\usepackage{listings}
\lstdefinestyle{customcpp}{
  breaklines=true,
  xleftmargin=\parindent,
  language=C++,
  showstringspaces=false,
  basicstyle=\small\ttfamily,
  keywordstyle=\bfseries\color{green!40!black},
  commentstyle=\itshape\color{purple!40!black},
  identifierstyle=\color{blue},
  stringstyle=\color{orange},
  tabsize=2,
  escapeinside={/@}{@/},
  captionpos=b,
  morekeywords={constexpr,decltype,real,interface},
}
\usepackage{siunitx}
\newlength\figureheight 
\newlength\figurewidth 

\usepackage[colorlinks=true]{hyperref}

\usepackage{doi}
\usepackage[square,numbers]{natbib}

\newcommand{\RR}{\mathbb{R}}
\newcommand{\CC}{\mathbb{C}}
\newcommand{\norm}[1]{\left\|#1\right\|}
\newcommand{\of}[1]{\left(#1\right)}
\newcommand{\intd}{\,\mathrm{d}}
\newcommand\elastic{\mathrm{e}}
\newcommand\reference{\mathrm{ref}}
\newcommand{\dev}{\operatorname{dev}}
\newcommand{\Div}{\operatorname{div}}
\newcommand{\T}{\mathrm{T}}
\newcommand{\viscop}[2][]{
\ifthenelse{\isempty{#1}}{#2_{vp}}{#2_{vp,\,#1}}
}
\newcommand{\misesnorm}[1]{\norm{#1}_{\text{eq}}}

\newcommand{\cpp}{C\texttt{++}\xspace}

\newcommand\Tstrut{\rule{0pt}{2.6ex}}
\newcommand\Bstrut{\rule[-0.9ex]{0pt}{0pt}}

\newcommand{\stepsize}{h}

\usepackage{afterpage}
\usepackage{lmodern}

\usepackage[normalem]{ulem}

\colorlet{lightblue}{blue!50}
\colorlet{lightred}{red!50}
\colorlet{lightgreen}{green!50}

\definecolor{emerald}{rgb}{0.31, 0.78, 0.47}
\definecolor{darkcyan}{rgb}{0.0, 0.55, 0.55}

\title{AutoMat --- Automatic Differentiation for Generalized Standard Materials on GPUs}
\author{Johannes Blühdorn \and
        Nicolas R. Gauger \and
	    Matthias Kabel}

\institute{Johannes Blühdorn (corresponding author)
			\at
           Chair for Scientific Computing \\
		   Technische Universität Kaiserslautern \\
		   \email{johannes.bluehdorn@scicomp.uni-kl.de}	
		   \and
		   Nicolas R. Gauger \at
           Chair for Scientific Computing \\
		   Technische Universität Kaiserslautern\\
		   \email{nicolas.gauger@scicomp.uni-kl.de}
		   \and
		   Matthias Kabel \at
		   Department of Flow and Material Simulation \\
		   Fraunhofer ITWM\\
		   \email{matthias.kabel@itwm.fraunhofer.de}}

\date{~}

\def\makeheadbox{\relax}

\begin{document}
\sloppy

\maketitle

\begin{abstract}
We propose a universal method for the evaluation of generalized standard materials that greatly simplifies the material law implementation process. By means of automatic differentiation and a numerical integration scheme, \emph{AutoMat} reduces the implementation effort to two potential functions. By moving AutoMat to the GPU, we close the performance gap to conventional evaluation routines and demonstrate in detail that the expression level reverse mode of automatic differentiation as well as its extension to second order derivatives can be applied inside CUDA kernels.
We underline the effectiveness and the applicability of AutoMat by integrating it into the FFT-based homogenization scheme of Moulinec and Suquet and discuss the benefits of using AutoMat with respect to runtime and solution accuracy for an elasto-viscoplastic example.
\keywords{Automatic Differentiation \and Generalized Standard Materials \and Numerical Methods for ODEs \and FFT-Based Homogenization \and GPU Computing}
\CRclass{G.1.4 \and G.1.7 \and G.4 \and J.2}
\end{abstract}

\section{Introduction}
\label{section:introduction}

In recent years, the improving quality of micro x-ray computed tomography (CT) images led to a digitalization of the material characterization process for composites. Nowadays, standard CT-devices have a maximum resolution below one $\mu m$ and produce $3$D images of up to $4096^3$ voxels. This permits a detailed view of the microstructure's geometry of composite materials up to the point where continuum approaches are still reasonable. In the context of material characterization, the physical description of the body leads to a partial differential equation (PDE) in which the behavior of the material itself is modeled in terms of a material law. Traditionally, a finite element (FEM) discretization is applied, and during the solution procedure, the material law is evaluated locally at quadrature points. To solve problems of this size with conventional FEM, large computing clusters are required to handle the global stiffness matrices \cite{Arbenz2008,Arbenz2014}.

In the last two decades, the FFT-based homogenization scheme of Moulinec and Suquet \cite{Moulinec1994,Moulinec1998} emerged as a memory efficient matrix-free alternative that was adapted to operate on structured finite element meshes \cite{Willot2015,Schneider2016,Schneider2017,Leuschner2018}. Besides the small memory footprint, the most favorable property of the so-called \emph{basic scheme} is a tangent-free treatment of nonlinear material behavior. However, its required iteration count is proportional to the material contrast, i.\,e.~the maximum of the quotient of the largest and the smallest eigenvalue of the algorithmic tangential stiffness field. Thus, for certain practical applications such as the homogenization of plastifying materials, the convergence behavior can be exceedingly slow \cite{Schneider2020}.

To accelerate the solution process, Zeman et al.~\cite{Zeman2010} and Brisard and Dormieux \cite{Brisard2010,Brisard2012} applied Krylov-subspace solvers to FFT-based homogenization. These methods are extremely fast, but they are restricted to linear problems. By combination with inexact Newton-methods, they were extended to the physically \cite{Gelebart2013} and geometrically \cite{Kabel2014} nonlinear case and exhibited excellent performance \cite{Lucarini2019,Ma2019}. The drawback of this approach consists in either loosing the small memory footprint or the need to calculate the tangential stiffness of the material laws in every iteration of the linear solver. Furthermore, the analytic derivation of the tangent can be tedious and its implementation may require considerable programming effort, and is thus prone to errors. This gave rise to applying Quasi-Newton methods in FFT-based micromechanics \cite{Shanthraj2014,Schneider2019,Chen2019,Chen2019b,Wicht2019,Schneider2020}. There, material tangents are replaced by suitable approximations. To sum up, the choice of the solver is driven by compromises between runtime efficiency, memory efficiency and the implementational effort of an accurate material tangent.

Especially during prototyping and modeling, it might be necessary to assess different material laws. Clearly, it is impractical to derive the material tangent from scratch for every material law under consideration. However, it is also undesirable to be restricted to tangent-free solvers during this phase. Motivated by the work of Rothe and Hartmann \cite{Rothe2014}, we started the development of \emph{AutoMat}, which leverages automatic differentiation and GPU computing to simultaneously address issues of flexibility, accuracy and performance.

Automatic differentiation (AD) refers to techniques for the automatic acquisition of machine accurate derivatives of computer codes \cite{Griewank2008}. 
These have applications in, e.\,g., the setup of adjoint solvers \cite{Schlenkrich2008}, parameter identification \cite{Auroux2017}, shape optimization \cite{Gauger2008}, and machine learning \cite{Guenther2020}. There, AD is applied to a full simulation. Here, we use AD locally for the automatic setup of solvers and eliminate the inconvenience of hand-computed derivatives. For classical CPU architectures, several mature AD tools are available as of now, for example ADOL-C \cite{Walther2009}, dco/c++ \cite{Leppkes2016} and CoDiPack \cite{Sagebaum2019}. Advances in the direction of AD for GPU codes are more recent, examples include dco/map  with applications in computational finance \cite{Leppkes2017}.
In \cite{Rothe2014}, Rothe and Hartmann use the source transformation tool OpenAD \cite{Utke2008} for the automatic computation of material tangents and the assembly of Jacobians for implicit solvers in the context of a multi-level Newton algorithm. In this work, the automatic differentiation ansatz is advanced in several directions.

We focus on the class of generalized standard materials (GSM) \cite{Halphen1975}, which we introduce in Section \ref{sec:GSM}. There, AD enables us to recover the constitutive equations of the material law automatically from given implementations of two potentials, resulting in a fully automatic solver setup. This allows for a highly usable and convenient integration of GSMs into mechanical solvers. We demonstrate this by integrating AutoMat into the FFT-based homogenization scheme of Moulinec and Suquet \cite{Moulinec1998} as implemented in FeelMath\footnote{\url{https://www.itwm.fraunhofer.de/feelmath}}. As our benchmark example for AutoMat, we use an elasto-viscoplastic material model with material parameters adjusted to measurements of a metal-matrix composite. The precise setup is taken from Michel and Suquet \cite{Michel2016} and summarized in Section \ref{section:example}.

The consistent tangent operator is the algorithmic derivative of the stress as it is computed from the strain according to the material law. Its computation requires a differentiation through an integration scheme for ordinary differential equations (ODEs). The conventional backward Euler step is differentiated in \cite{Simo2006} by hand. We show in Section \ref{subsec:single_implicit_euler_step} that this procedure can be fully automatized. To understand the numerical properties of the tangent computation, we interpret it in Section \ref{subsection:rosenbrock_runge_kutta} as a single implicit Euler step applied to an ODE for the derivative. Since this ODE depends on the chosen loading step size, convergence of the tangent for decreasing step size is not guaranteed. This motivated us to explore schemes with adaptive time steps instead.
The differentiation of ODE integration schemes in a blackbox manner, that is, without consideration of the structure of the integration algorithm and its approximative nature, usually leads to incorrect derivatives \cite{Eberhard1999}.
Therefore, we refine the strategy of solving simultaneously an ODE for the derivatives \cite{Eberhard1999,Pruess2010} in the presence of step size control for Rosenbrock methods and both explicit and implicit Runge-Kutta schemes.
Particularly, the relation to blackbox differentiation is explored. With the results, we can guarantee that the tangent is as accurate as the primal solution. The overall robustness and accuracy of the proposed scheme is assessed in Section \ref{subsection:solution_accuracy}. We achieve further robustness with respect to the choice of the ODE solver by employing a stress-driven error control, which we present in Section \ref{subsec:stress_driven_error_control}.

In FFT-based homogenization, the computationally costly simulation components are the Fourier transform and material law evaluation \cite{GrimmStrele2019}. For nonlinear materials, the latter tends to dominate the overall run time \cite{Kabel2017} and is hence performance critical. The spatial independence of material law evaluations allows for parallelization, which is typically used in an efficient implementation.
Throughout Section \ref{section:automatic_evaluation}, we compare parallelized material law evaluations on the CPU with GPU accelerated material law evaluation. We achieve a notable speedup for conventional material law evaluation, but particularly for the computationally more involved automatic evaluation strategies presented in this paper, there are significant performance gains. In our example and setup, we were able to close the performance gap between conventional material law evaluation on the CPU and automatic material law evaluation on the GPU.
The good performance would not be possible without an efficient implementation of automatic differentiation on the GPU. Therefore, we developed an operator overloading AD tool specifically for the application presented in this paper. It is based on expression template techniques; previously in AD, these were successfully applied for the treatment of right hand sides in the forward mode \cite{Phipps2012} and in Jacobi taping \cite{Hogan2014} as well as primal value taping \cite{Sagebaum2018} in the reverse mode.
The details of the implementation and its further optimizations are presented in Section \ref{section:automatic_differentiation_on_gpus}.
In Section \ref{section:layout_profiling_limits}, remaining influence factors on the performance are discussed. We analyze the performance limiters of AutoMat, present design choices and optimizations of the GPU implementation and discuss overlap of CPU workloads, GPU workloads, and data exchange as well as reductions of the memory footprint.

Finally, we summarize and conclude our work in Section \ref{sec:Conclusion}.

\section{Generalized Standard Materials}
\label{sec:GSM}

The notion of generalized standard materials is originally introduced in \cite{Halphen1975}; a compact introduction to the subject can be found in \cite{Michel2016}. Let $\varepsilon$ denote the right Cauchy-Green strain tensor, $\sigma$ the Cauchy stress tensor and $a\in\RR^m$ the vector of internal variables, all depending on time and space. The constitutive equations of the material law are given in terms of a Helmholtz free energy density $(\varepsilon,a)\mapsto\omega(\varepsilon,a)$ and a force potential $A\mapsto\Psi(A)$ and read
\begin{align}
	\label{eq:gsm_sigma}
	\sigma&=\frac{\partial\omega}{\partial\varepsilon}(\varepsilon,a),\\
	\label{eq:gsm_internal}
	\dot{a}&=\frac{\partial\Psi}{\partial A}\of{-\frac{\partial\omega}{\partial a}(\varepsilon,a)}.
\end{align}
$A$ is referred to as generalized stresses and if both $\omega$ and $\Psi$ are convex functions of their arguments, we speak of a generalized standard material. The dissipation potential which is the convex dual of $\Psi$ is not used in the present study.

After space discretization, evaluations of above stress-strain relationship and evolution of internal variables are required in the quadrature points. We drop the $x$ dependency in the notation as the specific location does not change throughout a single material law evaluation. After time discretization, the material law inputs at a quadrature point consist of a strain tensor $\varepsilon_n$ and internal variables $a_n$ at time $t_n$ as well as a strain tensor $\varepsilon_{n+1}$ which is usually only a prediction of the actual strain tensor at time $t_{n+1}$ in the context of the surrounding elasticity solver.
Then, in each quadrature point, the material law can be evaluated as follows.
\begin{enumerate}
	\item Solve the ODE for the internal variables \eqref{eq:gsm_internal} with initial data $(t_n,a_n)$ on the time interval $[t_n,t_{n+1}]$. Recover $\varepsilon(t)$ by means of linear interpolation between $\varepsilon_n$ and $\varepsilon_{n+1}$. This way, obtain $a_{n+1}$.
	\item Compute $\sigma_{n+1}$ via \eqref{eq:gsm_sigma} from $\varepsilon_{n+1}$ and $a_{n+1}$.
\end{enumerate}
Additionally, the consistent tangent operator $C_{n+1}$ which is the algorithmic derivative of $\sigma_{n+1}$ with respect to $\varepsilon_{n+1}$ is usually computed along with the material law. It is used in the FFT-based homogenization scheme to determine the optimal reference material.

In view of the decision for an integration scheme for \eqref{eq:gsm_internal}, negative eigenvalues of the Jacobian of the ODE's right hand side indicate that explicit solvers might display unstable behaviour \cite{Hairer2010}, that is, require extremely small steps. The following  theorem states that evolution equations arising from GSMs are subject to this issue. Following \cite{Wu1988}, $M\in\CC^{m\times m}$ is called positive semi-definite if $\forall x\in\CC^m\colon x^\ast Mx\in\RR_{\geq 0}$. This definition implies that each positive semi-definite complex matrix is Hermitean.

\begin{theorem}
	\label{theorem:gsm_jacobian_rhs_eigenvalues}
	Let a GSM be specified by $\omega$ and $\Psi$ and assume that both are $C^2$. If $\lambda$ is an eigenvalue of the Jacobian with respect to $a$ of the right hand side of \eqref{eq:gsm_internal}, then $\lambda\in\RR_{\leq 0}$.
\end{theorem}

\begin{proof}
	The Jacobian with respect to $a$ of the right hand side of \eqref{eq:gsm_internal} reads
	\begin{multline}
		\label{eq:gsm_jacobian_rhs}
			\frac{\intd}{\intd a}\of{\frac{\partial\Psi}{\partial A}\of{-\frac{\partial\omega}{\partial a}(\varepsilon, a)}}\\
			=-\frac{\partial^2\Psi}{\partial A^2}\of{-\frac{\partial\omega}{\partial a}(\varepsilon, a)}\frac{\partial^2\omega}{\partial a^2}(\varepsilon, a).
	\end{multline}
	As Hessians of $C^2$ functions, both $\frac{\partial^2\Psi}{\partial A^2}$ and $\frac{\partial^2\omega}{\partial a^2}$ are symmetric. Since both $\omega$ and $\Psi$ are convex and $C^2$, $\frac{\partial^2\Psi}{\partial A^2}$ and $\frac{\partial^2\omega}{\partial a^2}$ are also positive semi-definite as real matrices, that is,
	\begin{equation*}
		\forall x\in\RR^m\colon x^\T Mx\geq 0,
	\end{equation*}
	where $M$ denotes any of both Hessians. Symmetric and positive semi-definite real matrices are also positive semi-definite as complex matrices. By Theorem 2.2 in \cite{Wu1988}, the product of positive semi-definite complex matrices is similar to a positive semi-definite complex matrix, that is, there exists an invertible complex matrix $T$ such that $T^{-1}\frac{\partial^2\Psi}{\partial A^2}\frac{\partial^2\omega}{\partial a^2}T$ is a positive semi-definite complex matrix. All eigenvalues of a positive semi-definite complex matrix lie in $\RR_{\geq 0}$. As similarity preserves eigenvalues, all eigenvalues of the product $\frac{\partial^2\Psi}{\partial A^2}\frac{\partial^2\omega}{\partial a^2}$ are contained in $\RR_{\geq 0}$; hence all eigenvalues of \eqref{eq:gsm_jacobian_rhs} are contained in $\RR_{\leq 0}$.\qed
\end{proof}

Another example for an eigenvalue proof based on definiteness and convexity in the context of material simulation can be found in \cite{Cormeau1975}. There, a time-marching scheme for the solution of a viscoplastic problem is identified as a system of ODEs for the stresses at integration points and the eigenvalues of the Jacobian of the right hand side are used to assess stability properties.

Back to Theorem \ref{theorem:gsm_jacobian_rhs_eigenvalues}, whether explicit solvers (with adaptive step size control) or implicit solvers are faster depends on the specific material law, internal variable values, applied strain and integration interval length. To give a short example, we ignore the elasticy solver and focus on the material law evaluation at a single voxel. Consider the ODE arising from the example \eqref{eq:michel_suquet_omega}, \eqref{eq:michel_suquet_psi} that is introduced in the next section with parameters from Table \ref{table:misu_parameters}. We set $\varepsilon_n=0$, $a_n=0$, $\varepsilon_{xx,\,n+1}=0.4\,\%$ and integrate over time intervals of varying length $\Delta t$. Figure \ref{figure:ode_solver_choice} displays the numbers of intermediate steps and run times observed with the ODE solvers available in MATLAB\footnote{\url{https://de.mathworks.com/products/matlab.html}}.
\begin{figure*}[t]
	\centering
	\begin{minipage}{0.4\textwidth}
		\includegraphics{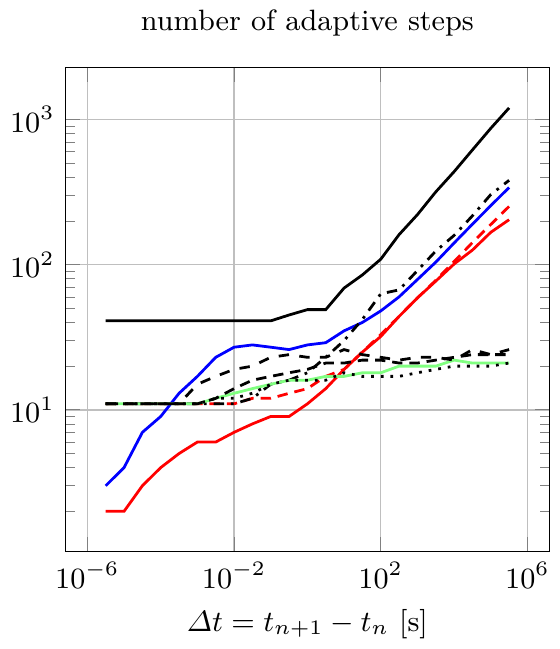}
	\end{minipage}
	\begin{minipage}{0.59\textwidth}
	    \includegraphics{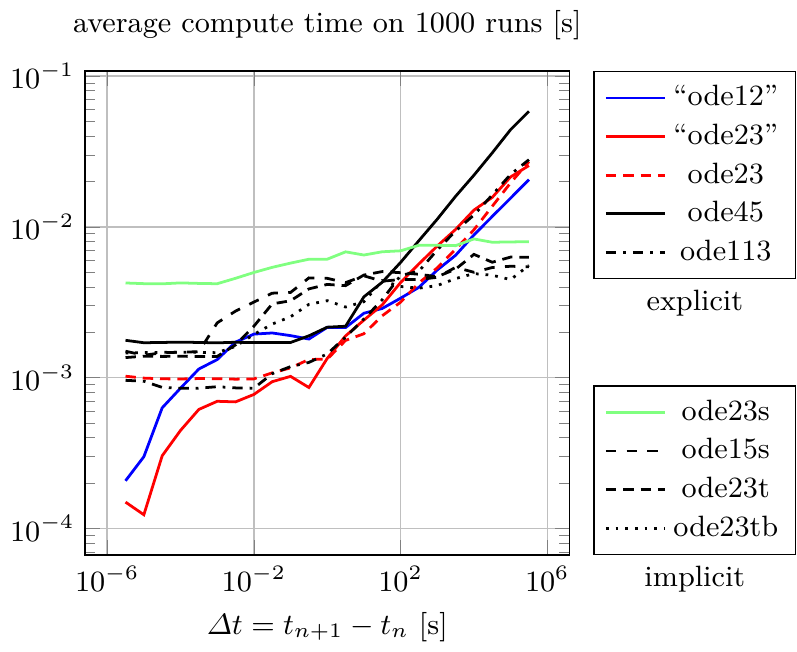}
	\end{minipage}
	\caption{Integration of ODE from elasto-viscoplastic GSM \cite{Michel2016} with MATLAB solvers and default tolerances. Includes also a lower order scheme ``ode12'' and for comparison an analogous custom implementation ``ode23''. Details on the schemes are included in Section \ref{subsection:rosenbrock_runge_kutta}.}
	\label{figure:ode_solver_choice}
\end{figure*}
The performance of explicit solvers is competitive up to rather large integration interval lengths. It is clearly linked to the number of intermediate steps taken by adaptive step size control and only for large $\Delta t$, the number of adaptive steps taken by explicit solvers is driven by stability rather than accuracy and increases with $\Delta t$.
In Section \ref{section:automatic_evaluation}, we refine both explicit and implicit solution strategies.

\section{Example}

\label{section:example}

Throughout the paper at hand, we perform our numerical studies for a uni-axial tension-compression test of a short fiber reinforced metal-matrix composite (MMC) taken from \cite{Michel2016}.

\paragraph{Microstructure}

The MMC consists of $10.2$ vol\% Al203 fibers embedded in an aluminum matrix. In our periodically generated micro-structure (see Figure \ref{fig:MMC}), the planar isotropic distributed fibers have a diameter of 9 $\mu$m and a length of 135 $\mu$m. This volume element of $150\times 150\times 150$ $\mu$m$^3$ was discretized by $150 \times 150 \times 150$ voxels.

\begin{figure}[H]
\includegraphics[width=0.48\textwidth]{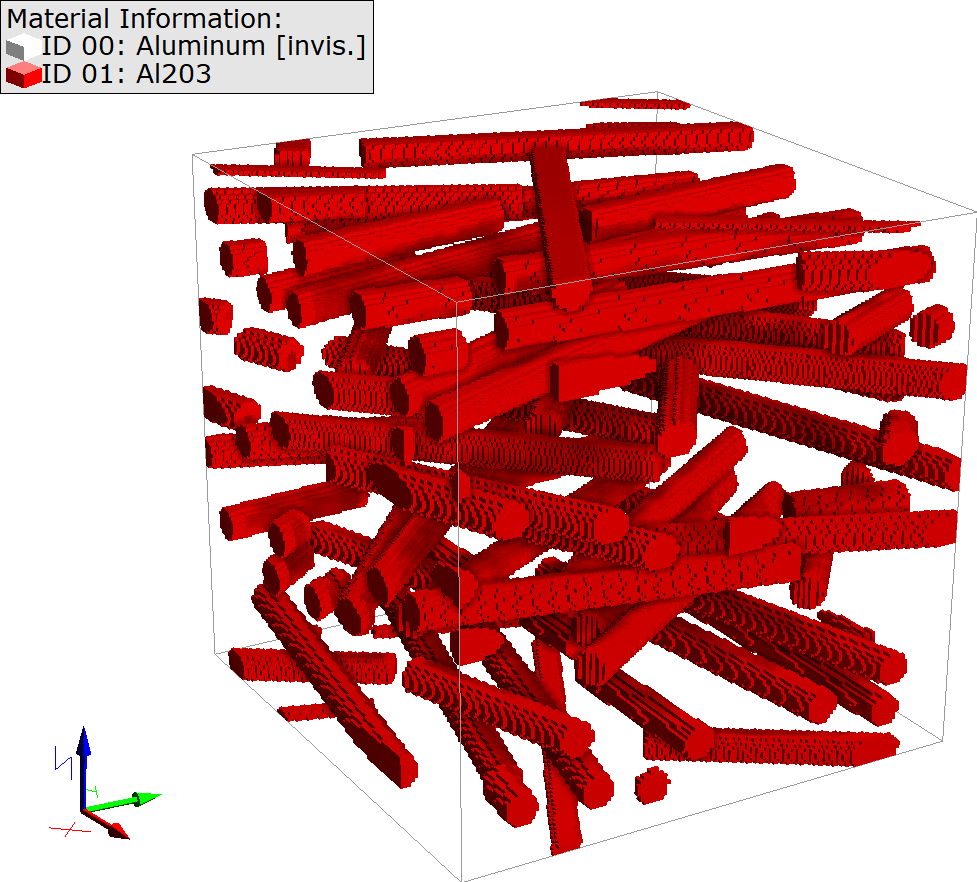}
\caption{Micro-structure of the MMC generated with GeoDict\protect\footnotemark.}
\label{fig:MMC}
\end{figure}
\footnotetext{\url{https://www.geodict.com/}}

\paragraph{Material Model}

The Al203 fibers are modeled linear elastic with Young's modulus $E$ and Poisson's ratio $\nu$ and the aluminum matrix as the elasto-viscoplastic GSM given by the potentials
\begin{multline}
	\label{eq:michel_suquet_omega}
	\omega(\varepsilon,\,\viscop\varepsilon,\,\alpha)=\\
	\frac{1}{2}\varepsilon_\elastic^T C^\elastic\varepsilon_\elastic+\frac{1}{3}\viscop\varepsilon^T\mathrm{H}\viscop\varepsilon+\int_0^\alpha \mathcal{K}(q)\intd q,
\end{multline}
where $\mathrm{H}=\operatorname{diag}\of{H,H,H,\frac{H}{2},\frac{H}{2},\frac{H}{2}}$ and $\varepsilon_\elastic=\varepsilon-\viscop\varepsilon$, and
\begin{multline}
	\label{eq:michel_suquet_psi}
	\Psi(\viscop{A},\,A_{\alpha})=\\
	\frac{\sigma_d \dot{\varepsilon}_0}{n+1}\of{\frac{\of{\misesnorm{\dev \viscop{A}}+A_{\alpha}}^+}{\sigma_d}}^{n+1}
\end{multline}
with viscoplastic strain $\viscop\varepsilon$ and equivalent plastic strain $\alpha$ as internal variables. $C^\elastic$ is an elastic stiffness matrix given in terms of a second $(E,\,\nu)$ pair and $\mathcal{K}(\alpha)$ describes the isotropic hardening and $H$ the (linear) kinematic hardening, whereas the viscous effects are given by the drag stress $\sigma_d$, the rate sensitivity $n$ and the reference strain rate $\dot{\varepsilon}_0$. For computational efficiency, the Voigt notation \cite{Voigt1966} is used for strain and stiffness tensors.

For the studied example, the nonlinear parameters of the aluminum matrix were calibrated without isotropic hardening, i.\,e.~$\mathcal{K}(\alpha)$ was assumed to be equal to the initial yield stress $\sigma_Y$, $\mathcal{K}(\alpha) \equiv \sigma_Y$.
The complete set of material parameters is reproduced in Table \ref{table:misu_parameters}.

\begin{table}
\centering
\begin{tabular}{c c | c c c@{}}
\toprule
\textbf{Parameter} & \textbf{Unit}  & \textbf{Aluminum} & \textbf{Al203} \tabularnewline
\midrule 
$E$ & GPa  & $55$  & $300$ \tabularnewline
$\nu$ & 1 & $0.33$ & $0.25$ \tabularnewline
\midrule
$\sigma_Y$ & MPa & $25$ & - \tabularnewline
$H$ & GPa & $1.8$ & - \tabularnewline
$\dot{\varepsilon}_0$ & 1/s & $1$ & - \tabularnewline
$\sigma_d$ & MPa & $130$ & - \tabularnewline
$n$ & 1 & $3.6$ & - \tabularnewline
\bottomrule
\end{tabular}
\caption{Parameters for elasto-viscoplastic GSM \protect\cite{Michel2016}.}
\label{table:misu_parameters}
\end{table}

\paragraph{Boundary Conditions}
As described in detail by Michel and Suquet, the volume element is submitted to a uni-axial tension-compression test at constant strain rate with alternating sign in loading direction (see Figure \ref{fig:BoundaryCondition}),
\begin{align*}
	&\dot{\varepsilon}_{xx} = \pm 1.4 \cdot 10^{-3}\,\textrm{s}^{-1},\\
	&-3.48441 \cdot 10^{-3} \leq \varepsilon_{xx} \leq 3.58454 \cdot 10^{-3}.
\end{align*}
The loading path is discretized in an equidistant manner with a granularity between 20 and 320 steps. \emph{If not mentioned otherwise, 80 loading steps are used.}
\begin{figure}[H]
\centering
\includegraphics{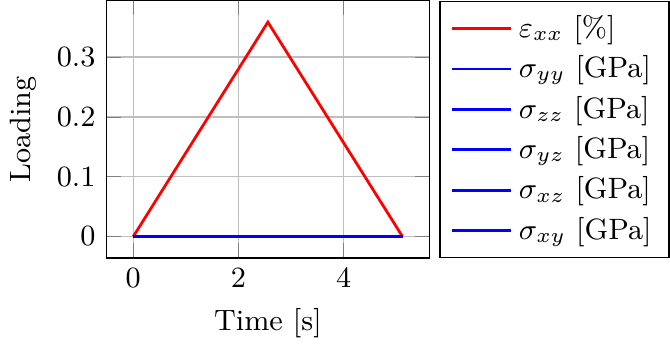}
\caption{Mixed boundary conditions \protect\cite{Kabel2016} for the uni-axial experiment.}
\label{fig:BoundaryCondition}
\end{figure}

\paragraph{Material Law Evaluations}

In each loading step, a stationary elastic problem is solved by FFT-based homogenization \cite{Moulinec1998}. This method is relying on an FFT-based preconditioner \cite{Kabel2014} defined by the constant coefficient linear elastic problem
$\Div\of{C^\reference \varepsilon} = 0$,
where $C^\reference$ is called the \emph{reference stiffness} and has to be chosen depending on the locally varying tangential stiffness of the material laws \cite{Kabel2014}. The reference stiffness can be either fixed at the beginning of the time dependent simulation by using only the initial elastic stiffness of the material laws or it can be adjusted in each loading step to the current tangential stiffness to reduce the number of iterations necessary for convergence. In the first case, this involves one material law evaluation per voxel \emph{with} tangent at the beginning of the initial loading step and in the latter case at the beginning of each loading step. The (matrix-free) FFT-based solver itself only performs one material law evaluation \emph{without} tangent per iteration and voxel. The performance impact of the reference material setup prior to the first loading step is negligible; therefore, \emph{whenever we display time spent on material law evaluations with tangent, the configuration at hand updates the reference material.} Then, material law evaluations with and without tangent are timed separately. We use the types of error control explained in Section \ref{subsec:stress_driven_error_control} throughout.

\paragraph{Parallelization}

We perform our tests on a dual-socket workstation with two Intel Xeon E5-2687Wv4 processors at 3\,GHz ($2\times 12$ cores) and an Nvidia Quadro GV100 graphics card. As this card has uncapped double precision performance, we keep the elasticity solver's double precision also for material law evaluations on the GPU. Nonetheless, single precision seems to work well for the material law presented above. This is of importance on GPUs without good double precision performance, and can also speed up computations in general; especially material law evaluations with tangent seem to benefit performance-wise from single precision. We use OpenMP\footnote{\url{https://www.openmp.org/}} for CPU parallelization; on the graphics card, CUDA\footnote{\url{https://developer.nvidia.com/cuda-zone}} is used. Details on the computational layout can be found in Section \ref{section:layout_profiling_limits}.

\section{Automatic Evaluation}
\label{section:automatic_evaluation}

Conventionally, efficient methods for the evaluation of specific material laws are derived by hand. For example, GSMs such as \eqref{eq:michel_suquet_omega}, \eqref{eq:michel_suquet_psi} are discretized in Chapter 3 of \cite{Simo2006} by means of a single backward Euler step. With the help of an explicit formula for the flow direction, the resulting nonlinear system of equations is reduced to a scalar equation that is then solved by Newton's method. For the computation of $C_{n+1}$, the derivative of the corresponding nonlinear equation solve is recovered in an implicit function theorem fashion. Numerical integration and algorithmic differentiation are both carried out by hand. We refer to this approach as \emph{conventional evaluation strategy} --- it is material law specific. For our performance studies, it serves as a baseline. In this work, we explore several flavours of the \emph{automatic evaluation strategy} depicted in Figure \ref{figure:automatic_evaluation_strategy}
\begin{figure*}[t]
	\centering
	\includegraphics{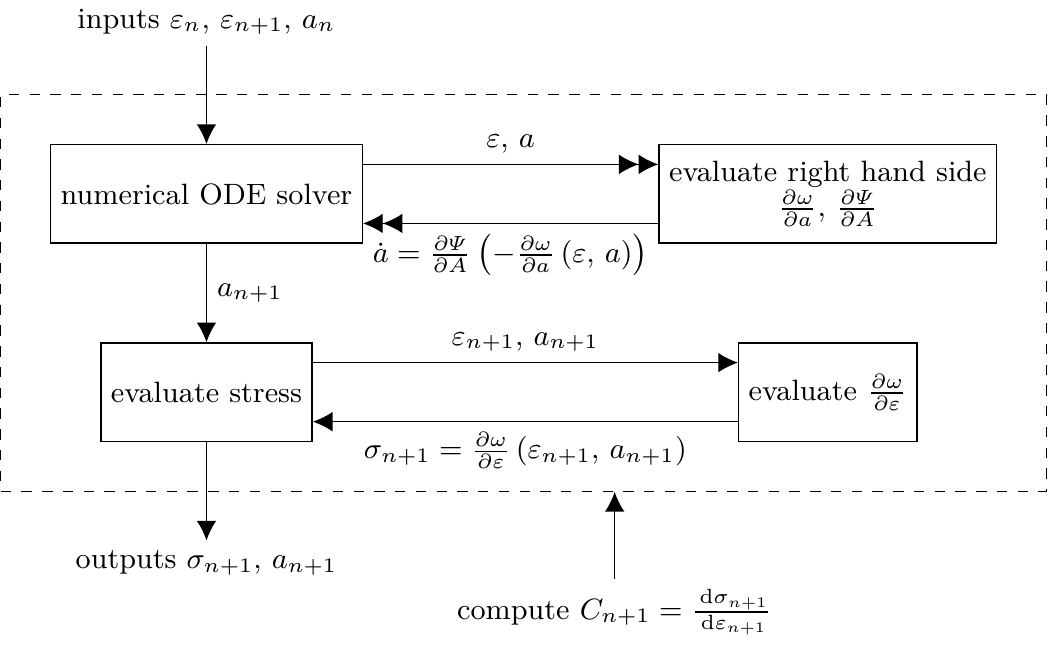}
	\caption{Automatic evaluation strategy.}
	\label{figure:automatic_evaluation_strategy}
\end{figure*}
that relies on AD to evaluate the various partials of $\omega$ and $\Psi$, to assemble Jacobians as required for ODE integration schemes and finally, to compute the material tangent $C_{n+1}$, which involves a differentiation of the whole algorithm depicted in Figure \ref{figure:automatic_evaluation_strategy}. The strategy can easily be adapted to other material laws by exchanging the implementations of the potentials. We also explore the performance benefits of providing hand-derived implementations of the partials of $\omega$ and $\Psi$ for an otherwise automatic evaluation; we refer to this as \emph{semi-automatic evaluation strategy}.

\subsection{Single Implicit Euler Step}
\label{subsec:single_implicit_euler_step}

AD allows us to turn the conventional scheme from \cite{Simo2006} into an automatic evaluation strategy that is not specific to a certain GSM and requires only implementations of $\omega$ and $\Psi$. Let $\stepsize=t_{n+1}-t_n$ be the loading step size and
\begin{equation*}
	f(\varepsilon,\,a)=\frac{\partial\Psi}{\partial A}\of{-\frac{\partial\omega}{\partial a}(\varepsilon,\,a)},
\end{equation*}
that is, the right hand side of the ODE \eqref{eq:gsm_internal}. An application of a single implicit Euler step yields 
\begin{equation*}
	a_{n+1}=a_n+\stepsize\cdot f(\varepsilon_{n+1},\,a_{n+1}),
\end{equation*}
that is, the nonlinear system of equations
\begin{equation}
	\label{eq:implicit_euler}
	\underbrace{a_{n+1}-\stepsize\cdot f\of{\varepsilon_{n+1},\,a_{n+1}}-a_n}_{=:\,F(\varepsilon_{n+1},\,a_{n+1})}=0
\end{equation}
for $a_{n+1}$, which we solve with Newton's method. We initialize $a_{n+1}^{(0)}=a_n$ and iterate
$a_{n+1}^{(k+1)}=a_{n+1}^{(k)}-\Delta a_{n+1}^{(k)}$ where
\small
\begin{multline*}
	\Delta a_{n+1}^{(k)}=\of{\frac{\partial F}{\partial a}\of{\varepsilon_{n+1},\,a_{n+1}^{(k)}}}^{-1}F\of{\varepsilon_{n+1},\,a_{n+1}^{(k)}}.
\end{multline*}
\normalsize
The application of AD is twofold. Each evaluation of $f$ (or $F$) involves evaluations of the partials $\frac{\partial\Psi}{\partial A}$ and $\frac{\partial\omega}{\partial a}$. This can be automatized by first order automatic differentiation. Second, the evaluations of the Jacobian $\frac{\partial F}{\partial a}$ can be realized likewise by AD but require --- due to the already involved partials --- an additional derivative order. Note $\frac{\partial F}{\partial a}=I-\stepsize\cdot\frac{\partial f}{\partial a}$, so it suffices to apply AD to $f$. The material tangent
\begin{multline}
	\label{eq:tangent_formula}
	C_{n+1}=\frac{\mathrm{d}\sigma_{n+1}}{\mathrm{d}\varepsilon_{n+1}}=\frac{\partial^2\omega}{\partial\varepsilon^2}(\varepsilon_{n+1},\,a_{n+1})\\
	+\frac{\partial^2\omega}{\partial a\partial\varepsilon}(\varepsilon_{n+1},\,a_{n+1})\frac{\mathrm{d} a_{n+1}}{\mathrm{d}\varepsilon_{n+1}}
\end{multline}
requires the derivative of the evolved internal variables with respect to the predicted strain. Assuming --- similar to the derivation of the scheme in \cite{Simo2006} --- that the primary system of equations was solved exactly, it holds by differentiating \eqref{eq:implicit_euler} with respect to a single strain component
\begin{multline*}
	\frac{\mathrm{d} a_{n+1}}{\mathrm{d}\varepsilon_{n+1,\,i}}=\stepsize\cdot\frac{\partial f}{\partial\varepsilon_{n+1,\,i}}(\varepsilon_{n+1},\,a_{n+1})\\
	+\stepsize\cdot\frac{\partial f}{\partial a}(\varepsilon_{n+1},\,a_{n+1})\frac{\mathrm{d} a_{n+1}}{\mathrm{d}\varepsilon_{n+1,i}},
\end{multline*}
\vspace{-2ex}
that is,\\
\vspace{-2ex}
\begin{multline*}
	\of{I-\stepsize\cdot\frac{\partial f}{\partial a}\of{\varepsilon_{n+1},\,a_{n+1}}}\frac{\mathrm{d} a_{n+1}}{\mathrm{d} \varepsilon_{n+1,\,i}}\\
	=\stepsize\cdot\frac{\partial f}{\partial\varepsilon_{n+1,\,i}}(\varepsilon_{n+1},\,a_{n+1}).
\end{multline*}
Hence, the required derivative values can be obtained in a postprocessing step by six additional linear system solves, one for each Voigt component of the strain and with the same coefficient matrix the next Newton iteration would use. $\frac{\partial f}{\partial\varepsilon_{n+1,\,i}}$ can be evaluated with AD analogously to $\frac{\partial f}{\partial a}$. Since $\frac{\mathrm{d}\varepsilon_{n+1}}{\mathrm{d}\varepsilon_{n+1}}=I$, it is then straightforward to propagate the derivatives with respect to the strain with AD through an evaluation of the stress relationship \eqref{eq:gsm_sigma} to obtain both $\sigma$ and $C_{n+1}$. As before, this involves also a partial of $\omega$ and requires second order AD capabilities.

Table \ref{table:implicit_euler_timings} provides an overview over time spent with the implicit Euler variants on material law evaluations with and without tangent in our running example. Here, all displayed configurations perform the exact same number of both types of material law evaluations, and the timings are immediately comparable. We should also mention that here, all simulation results obtained are identical up to machine precision.
\begin{table*}
\centering
\begin{tabular}{c c c c@{}}
\toprule
\textbf{architecture} & \textbf{evaluation strategy} & \makecell{\textbf{material law [s]}\\no tangent} & \makecell{\textbf{material law [s]}\\tangent}\tabularnewline
\midrule
CPU & conventional & 830.1 & 50.9\tabularnewline
CPU & automatic & 4996.7 & 408.1\tabularnewline
CPU & semi-automatic & 1668.1 & 127.5\tabularnewline
\midrule
GPU & conventional & 249.7 & 34.9\tabularnewline
GPU & automatic & 257.7 & 38.1\tabularnewline
GPU &  semi-automatic & 257.8 & 38.3\tabularnewline
\bottomrule
\end{tabular}
\caption{Total time spent on both types of material law evaluations with implicit Euler strategies.}
\label{table:implicit_euler_timings}
\end{table*}
The timings reveal two important trends. First, the automatization on the CPU is costly. Given the significant runtime improvements from switching to the semi-automatic evaluation strategy, part of this cost is due to AD and the automatic computation of the partials of $\omega$ and $\Psi$. Another part of the cost is due to the generality. Other than in the conventional implementation, we lack additional knowledge about the roles of internal variables. We have no formula for the flow direction and solve a full system of nonlinear equations with Newton's method instead. All evaluation strategies are notably accelerated by the GPU, and here, most important, \emph{even keeping the full automatization does not incur visible performance costs}. This is due to overlap of CPU and GPU workloads as detailed in Section \ref{section:layout_profiling_limits}.

\subsection{Rosenbrock and Runge-Kutta Schemes with Adaptive Step Size}
\label{subsection:rosenbrock_runge_kutta}

With a single implicit Euler step, there is no direct form of error control for the involved material law evaluations. The surrounding elasticity solver cannot compensate this lack of accuracy and will therefore solve the time discretized elasticity problem with potentially wrong stress (and stiffness) input. This regards nonlinear effects in particular. Since we cannot know in advance if and when these take place, we have to discretize the whole loading path with small loading steps. As we detail in the following, while this helps with the accuracy of stresses, the accuracy of tangents can not necessarily be guaranteed this way.

To that end, we first establish an interpretation of the algorithmic derivative of a single implicit Euler step as introduced in the previous section as a single implicit Euler step applied to an ODE for the derivative. Let a parameter dependent ODE system
\begin{equation}
	\label{eq:ode_with_parameters}
	\dot{y}=f(y,\,p)
\end{equation}
be given. We differentiate both sides of \eqref{eq:ode_with_parameters} with respect to $p$ and formally interchange the order of derivatives on the left hand side to obtain
\begin{equation}
	\label{eq:differentiated_ode_with_parameters}
	\frac{\intd}{\intd t}\of{\frac{\intd y}{\intd p}}=\frac{\partial f}{\partial y}(y,\,p)\frac{\intd y}{\intd p}+\frac{\partial f}{\partial p}(y,\,p).
\end{equation}
Assuming sufficient smoothness \cite{Pruess2010}, the derivative of $y$ with respect to $p$ is the unique solution to \eqref{eq:differentiated_ode_with_parameters} together with an initial value. The implicit Euler scheme with step size $\stepsize$ applied to the coupled system formed by \eqref{eq:ode_with_parameters} and \eqref{eq:differentiated_ode_with_parameters} yields
\begin{equation}
	\label{eq:euler_primal}
	y_{n+1}=y_n+\stepsize f(y_{n+1},\,p),
\end{equation}
\vspace{-4ex}
\begin{multline}
	\label{eq:euler_derivative}
	\frac{\intd y_{n+1}}{\intd p}=\frac{\intd y_n}{\intd p}+\stepsize\frac{\partial f}{\partial y}(y_{n+1},\,p)\frac{\intd y_{n+1}}{\intd p}\\
	+\,\stepsize\frac{\partial f}{\partial p}(y_{n+1},\,p).
\end{multline}
Clearly, this can be solved in two stages. After a solve of the nonlinear equation \eqref{eq:euler_primal} for $y_{n+1}$, one linear solve of \eqref{eq:euler_derivative} is sufficient to recover the derivative $\frac{\intd y_{n+1}}{\intd p}$. However, \eqref{eq:euler_derivative} can equivalently be obtained in an algorithmic manner by differentiating \eqref{eq:euler_primal} with respect to $p$ as long as $\frac{\intd h}{\intd p}=0$. Hence, the algorithmic derivative of a single implicit Euler step has an interpretation as a single implicit Euler step applied to the ODE for the derivative.

This holds likewise for the single implicit Euler step applied in the schemes in Section \ref{subsec:single_implicit_euler_step} where we have already seen the two-step solution procedure. Now we deduce properties of the numerical tangent approximation via the ODE it approximates. Let $f(\varepsilon,\,a)=\frac{\partial\Psi}{\partial A}\of{-\frac{\partial\omega}{\partial a}(\varepsilon,\,a)}$ denote the right hand side of the evolution equation \eqref{eq:gsm_internal}. In the setting of Section \ref{subsec:single_implicit_euler_step}, we have to consider the numerical ODE solve in the context of a single material law evaluation with initial data $a_n$ and $\frac{\intd a_n}{\intd\varepsilon_{n+1}}=0$ with step size $\stepsize = t_{n+1}-t_n$ over the time interval $[t_n,\,t_{n+1}]$. Here, $\varepsilon_{n+1}$ plays the role of the parameters. The evolution equation $\dot{a}=f(\varepsilon(t),\,a)$ leads to the ODE
\begin{multline}
	\label{eq:problematic_ode}
	\frac{\intd}{\intd t}\of{\frac{\intd a}{\intd\varepsilon_{n+1}}}=\frac{\partial f}{\partial a}(\varepsilon(t),\,a)\frac{\intd a}{\intd\varepsilon_{n+1}}\\
	+\frac{\partial f}{\partial\varepsilon}(\varepsilon(t),\,a)\frac{\intd\varepsilon(t)}{\intd\varepsilon_{n+1}}
\end{multline}
for the derivative. By the properties of the implicit Euler scheme, the numerical solve of the undifferentiated evolution equation is guaranteed to converge with order one to the exact solution as $\stepsize\to 0$. Here, the user can influence accuracy by choosing smaller loading steps.

For the ODE for the derivative \eqref{eq:problematic_ode}, the situation is different. Independent of the loading step size, $\varepsilon_{n+1}$ always refers to the strain value at time $t_{n+1}$. The linear interpolation
\begin{equation*}
\varepsilon(t)=\varepsilon_n\cdot\frac{t_{n+1}-t}{\stepsize}+\varepsilon_{n+1}\cdot\frac{t-t_n}{\stepsize}
\end{equation*}
between the known strain values leads to
\begin{equation*}
\frac{\intd\varepsilon(t)}{\intd\varepsilon_{n+1}}=\frac{t-t_n}{\stepsize},
\end{equation*}
which is the linear interpolation between 0 and 1 over the integration interval $[t_n,\,t_{n+1}]$. Therefore, the ODE for this particular derivative changes its shape with $\stepsize$. As the ODE is not invariant with respect to the integration interval, we cannot expect convergence to the exact solution with $\stepsize\to 0$ if only a single implicit Euler step is applied.

The following example illustrates that the relative error in the differentiated internal variables might even increase for $\stepsize\to 0$. We compare the results obtained by single implicit Euler steps to the results obtained by implicit Euler with a simple step size control mechanism. Consider a single voxel of the elasto-viscoplastic material \eqref{eq:michel_suquet_omega}, \eqref{eq:michel_suquet_psi} with the parameters from Table \ref{table:misu_parameters}. We use the mixed boundary conditions from Figure \ref{fig:BoundaryCondition}. This loading path is discretized by varying numbers of equidistant loading steps. For each loading step, a material law evaluation with or without substeps is performed. The relative errors observed in the derivative $\frac{\intd \viscop[n+1,\,xx]{\varepsilon}}{\intd \varepsilon_{n+1,\,xx}}$ can be seen in Figure \ref{fig:relative_error_tangent}. Clearly, the relative error increases for $\stepsize\to 0$.

\begin{figure*}
	\begin{center}
		\includegraphics{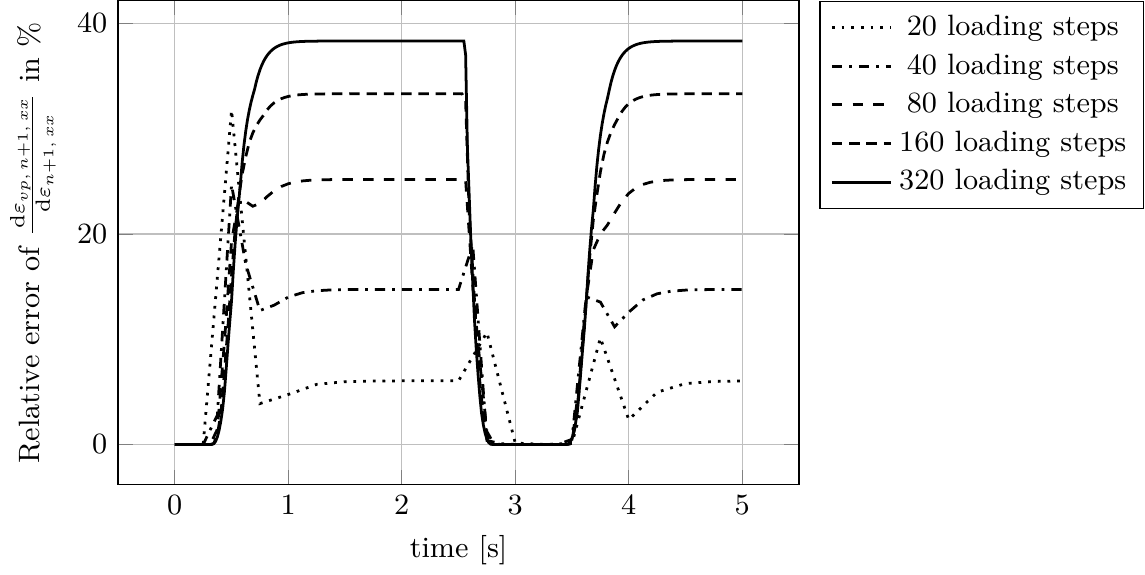}
	\end{center}
	\caption{Influence of loading step size on the relative error of a differentiated internal variable. Solution obtained by single implicit Euler steps compared to solution obtained by implicit Euler with adaptive substeps.}
	\label{fig:relative_error_tangent}
\end{figure*}

This shows that an accurate tangent evaluation cannot be performed without further discretization of the integration interval $[t_n,\,t_{n+1}]$ and serves as an additional motivation for adaptive substeps that are otherwise studied e.\,g.~in \cite{Arya1986} in the context of material law evaluation. Specifically, the material law inputs and outputs still follow the global time discretization, but locally, each material law evaluation uses a further discretization of $[t_n,\,t_{n+1}]$ to meet specified tolerances. In this section, we analyze well-known integration schemes with respect to automatic differentiation in the presence of step size control. \emph{Note that implicit Euler with adaptive steps is not used in the remaining parts of this paper; instead, schemes with step size control via an embedded method are considered.}

For adaptive time step sizes, the computation of the material tangent still requires the derivative of the evolved internal variables with respect to the predicted strain. Even if it is in principle possible to propagate those derivatives by AD through multiple steps of an ODE integration scheme in a blackbox manner, this corresponds to an algorithmic differentiation of an approximation and comprises a risk of inaccurate derivatives. The issues of blackbox differentiation of ODE integration schemes and possible solutions are discussed in \cite{Eberhard1999}. Particularly, two problems are mentioned. First, the step size is solely determined by the integration of the primal equation. Hence, there are no guarantees for the accuracy of the derivatives. Second, the differentiation of the step size control mechanism spoils the result with discretization dependent components. In \cite{Eberhard1999}, the focus is on an aposteriori error correction that recovers the desired derivatives from quantities obtained by blackbox differentiation. Here, we study the continuous approach to the problem in greater detail and refine the strategy of solving simultaneously an ODE for the derivative for the case of Rosenbrock methods and both explicit and implicit Runge-Kutta schemes in the presence of step size control. In Theorem \ref{theorem:ad_integration_step} and Corollary \ref{corollary:ad_integration_scheme}, we show that the ansatz is equivalent to suitably modified blackbox differentiation. Particularly, we guarantee that the derivatives are as accurate as the primal solutions.

For the sake of notational simplicity, we develop the following theory for autonomous ODEs and require implicitly that the used integration schemes satisfy the consistency condition that they yield the same numerical solution before and after transformation of the ODE to autonomous form.

Assuming sufficient smoothness \cite{Pruess2010}, the derivative of $y$ with respect to $p$ is the unique solution to \eqref{eq:differentiated_ode_with_parameters}.
The combined system \eqref{eq:ode_with_parameters} and \eqref{eq:differentiated_ode_with_parameters} inherits the stability properties of \eqref{eq:ode_with_parameters} in the sense that the Jacobian of the right hand side with respect to the unknowns is of block type
\begin{multline}
	\label{eq:combined_jacobian}
	\frac{\partial}{\partial\begin{bmatrix}y & \frac{\intd y}{\intd p}\end{bmatrix}}\begin{bmatrix}f(y,\, p) \\ \frac{\partial f}{\partial y}(y,\,p)\frac{\intd y}{\intd p}+\frac{\partial f}{\partial p}(y,\,p)\end{bmatrix}\\
	=\begin{bmatrix} \frac{\partial f}{\partial y}(y,\,p) & 0 \\ \ast & \frac{\partial f}{\partial y}(y,\,p)\end{bmatrix}
\end{multline}
and has the same eigenvalues as $\frac{\partial f}{\partial y}(y,\,p)$.

For some classes of integration schemes, the simultaneous solve of \eqref{eq:ode_with_parameters} and \eqref{eq:differentiated_ode_with_parameters} can be realized by means of automatically differentiating the numerical solve of \eqref{eq:ode_with_parameters} with respect to $p$ in a blackbox manner with some additional adaptions.
Let an integration scheme with $s$ stages and both linear and nonlinear implicit terms be specified by the update relations
\begin{align}
	\label{eq:general_stage_update}
	&k_i=\stepsize f\of{y_n^{(i)},\,p}+\stepsize J\sum_{j=1}^s\gamma_{ij}k_j,\\
	\label{eq:general_solution_update}
	&y_{n+1}=y_n+\sum_{j=1}^sb_jk_j
\end{align}
where $J=\frac{\partial f}{\partial y}(y_n,\,p)$ is the Jacobian of the right hand side and
\begin{equation*}
	y_n^{(i)}=y_n+\sum_{j=1}^sa_{ij}k_j.
\end{equation*}

\begin{theorem}
	\label{theorem:ad_integration_step}
	Let initial data $y_n$ and $\frac{\intd y_n}{\intd p}$ be given. The algorithmic derivative of a single step of the scheme \eqref{eq:general_stage_update}, \eqref{eq:general_solution_update} with step size $\stepsize$ applied to \eqref{eq:ode_with_parameters} yields the same value $\frac{\intd y_{n+1}}{\intd p}$ as an application of the same integration step to the combined system \eqref{eq:ode_with_parameters} and \eqref{eq:differentiated_ode_with_parameters} as long as $\frac{\intd \stepsize}{\intd p}=0$ and as long as the derivatives of equation solves are recovered according to the implicit function theorem. In terms of automatic differentiation, it is sufficient if $\stepsize$ does not carry derivative values and equation solves are treated as elementary operations.
\end{theorem}

\begin{proof}
For notational simplicity, let $p$ be scalar. Let $y$ denote the solution to \eqref{eq:ode_with_parameters}, $\frac{\intd y}{\intd p}$ its algorithmic derivative with respect to $p$ and $\begin{bmatrix}y & \tilde{y}\end{bmatrix}$ the solution to the combined system \eqref{eq:ode_with_parameters} and \eqref{eq:differentiated_ode_with_parameters}. Likewise, we refer to the stage vectors for the solution step of \eqref{eq:ode_with_parameters} as $k_i$ and to the stage vectors for the solution step of the combined system as $\begin{bmatrix}k_i & \tilde{k}_i\end{bmatrix}$. 
By the linearity of \eqref{eq:general_solution_update} and the initial value relation $\tilde{y}_n=\frac{\intd y_n}{\intd p}$, it is sufficient to ensure that $\tilde{k}_i=\frac{\intd k_i}{\intd p}$, $i=1,\,\ldots,\,s$. If we apply the integration step to the combined ODEs \eqref{eq:ode_with_parameters} and \eqref{eq:differentiated_ode_with_parameters}, the equations for the stage vector components $\tilde{k}_i$ read
\begin{multline}
	\label{eq:derivative_stage_vector_components}
	\tilde{k}_i=\stepsize\frac{\partial f}{\partial y}\of{y_n^{(i)},\,p}\tilde{y}_n^{(i)}+\stepsize\frac{\partial f}{\partial p}\of{y_n^{(i)},\,p}\\
	+\stepsize\tilde{J}\sum_{j=1}^s\gamma_{ij}k_j+\stepsize J\sum_{j=1}^s\gamma_{ij}\tilde{k}_j,
\end{multline}
where
\begin{equation*}
\tilde{J}=\frac{\partial}{\partial y}\of{\frac{\partial f}{\partial y}(y_n,p)\tilde{y}_n+\frac{\partial f}{\partial p}(y_n,\,p)}
\end{equation*}
is the lower left block of \eqref{eq:combined_jacobian} evaluated at $y_n$, $\tilde{y}_n$ and $p$.
However, as long as $\frac{\intd h}{\intd p}=0$, the same system of equations is obtained if we differentiate both sides of \eqref{eq:general_stage_update} with respect to $p$ and identify $\tilde{k}_i=\frac{\intd k_i}{\intd p}$. To that end, note $\tilde{J}=\frac{\intd J}{\intd p}$.
Hence, if we recover the algorithmic derivative of the $k_i$ 
from solves of the equations obtained by implicit differentiation, we obtain the same result as by solving an ODE for the derivative.\qed
\end{proof}

In the case of prescribed step sizes, Theorem \ref{theorem:ad_integration_step} extends inductively to multiple subsequent integration  steps. In the case of automatic step size control, for example via an embedded method according to \cite{Hairer1993}, the same holds true after small additional modifications.
\begin{enumerate}
	\item To meet the assumption $\frac{\intd \stepsize}{\intd p}=0$ of Theorem \ref{theorem:ad_integration_step} in terms of AD, the step size control mechanism must remain undifferentiated.
	\item To achieve the same accuracy for the solution components $y$ and $\frac{\intd y}{\intd p}$, all of them must be regarded in the step size control error measure.
\end{enumerate}
These additional modifications can also be found among the general suggestions in \cite{Eberhard1999}. Here, we have shown that they are --- together with the appropriate treatment of equation solves --- sufficient to turn blackbox differentiation of an ODE integration scheme of the type \eqref{eq:general_stage_update}, \eqref{eq:general_solution_update} into an algorithm that is equivalent to solving an ODE for the derivative.

\begin{corollary}
	\label{corollary:ad_integration_scheme}
	Theorem \ref{theorem:ad_integration_step} generalizes to subsequent integration steps also in the presence of automatic step size control as long as step size control is excluded from differentiation and derivative components are regarded in the error measure. The obtained derivative is as accurate as the primal solution.
\end{corollary}

Theorem \ref{theorem:ad_integration_step} and Corollary \ref{corollary:ad_integration_scheme} cover various classes of well-known integration schemes. If we choose $a_{ij}=0$ for $j\geq i$ and $\gamma_{ij}=0$ for $j>i$, \eqref{eq:general_stage_update} and \eqref{eq:general_solution_update} turn into a Rosenbrock scheme \cite{Hairer2010}. There, only linear implicit terms are used and \eqref{eq:derivative_stage_vector_components} can be simplified to $s$ linear solves
\begin{multline}
	\label{eq:rosenbrock_combined_update}
	(I-\gamma_{ii}\stepsize J)\tilde{k}_i=\stepsize\frac{\partial f}{\partial y}\of{y_n^{(i)},\,p}\tilde{y}_n^{(i)}\\
	+\stepsize\frac{\partial f}{\partial p}\of{y_n^{(i)},\,p}+\stepsize\tilde{J}\sum_{j=1}^i\gamma_{ij}k_j+\stepsize J\sum_{j=1}^{i-1}\gamma_{ij}\tilde{k}_j.
\end{multline}
The solve for $\tilde{k}_i$ can be performed immediately after the solve for $k_i$. For the choice $\gamma_{ij}=0$ for all $i$ and $j$, we obtain an implicit Runge-Kutta scheme \cite{Hairer1993}. The implicit Euler step discussed at the beginning of this section is an example for this and hence a special instance of Theorem \ref{theorem:ad_integration_step}. If additionally $a_{ij}=0$ for $j\geq i$, we obtain an explicit Runge-Kutta scheme \cite{Hairer1993}. There, no equation solves are required and Theorem \ref{theorem:ad_integration_step} simplifies to a straightforward application of forward AD to the stage vector updates. Otherwise, AD can be used to compute the derivatives required in the setup of \eqref{eq:derivative_stage_vector_components}.

In the GSM context, the components of $\varepsilon_ {n+1}$ play the role of the parameter $p$, $a_{n+1}$ corresponds to $y$ and $f$ is the right hand side of \eqref{eq:gsm_internal}. We apply Corollary \ref{corollary:ad_integration_scheme} for the computation of $\frac{\mathrm{d}a_{n+1}}{\mathrm{d}\varepsilon_{n+1}}$. For each class of integration schemes, the AD tool must be capable of computing various higher order derivatives. For explicit Runge-Kutta schemes, as before, we need one derivative order for the computation of the material tangent and one for the evaluation of the partials. For Rosenbrock methods, however, the computation of the Jacobian of the right hand side requires an additional derivative order. This is due to the term $\tilde{J}=\frac{\mathrm{d}J}{\mathrm{d}p}=\frac{\intd}{\intd p}\frac{\partial}{\partial a}\of{\frac{\partial\Psi}{\partial A}\of{\dots}}$ in \eqref{eq:rosenbrock_combined_update}. It is in principle possible to extend the AD tool presented Section \ref{section:automatic_differentiation_on_gpus} to third order derivatives. However, additional derivative orders incur an exponential increase in memory and/or runtime \cite{Griewank2008} and we do not expect reasonable performance. Thus, to recover one derivative order, the user has to implement the partials of $\omega$ and $\Psi$ explicitly in this case, i.\,e.~only the semi-automatic evaluation strategy is available.

We consider the pair of explicit Runge-Kutta schemes from \cite{Bogacki1989} that is known from MATLAB's routine ode23 and a lower-order Runge-Kutta pair formed by the explicit Euler scheme and Heun's method. This pair is also used for DAE integration in the context of material law evaluation in \cite{Hiley2008} and we refer to it as ode12. Finally, we include the Rosenbrock scheme from \cite{Shampine1997} that is behind MATLAB's ode23s. We implement all three with automatic step size control according to \cite{Hairer1993} and keep the MATLAB default tolerances $a_\text{tol}=10^{-6}$ and $r_\text{tol}=10^{-3}$. If we solve additionally for the derivatives, the solutions for the derivative of $a$ with respect to $\varepsilon_{n+1}$ enter the error measure in the same way as primal solution components.

Table \ref{table:advanced_integrators_timings} displays the timings for Runge-Kutta and Rosenbrock evaluation strategies. Compared to the previous timings in Table \ref{table:implicit_euler_timings} without adaptive step size control, we take notice that on the CPU, semi-automatic evaluations without tangent with ode12 and especially ode23 can be performed even faster than the conventional evaluation strategy. Often, one or a few adaptive steps are sufficient, and Runge-Kutta steps are computationally cheaper than those of implicit schemes since no equation solves are involved. Note that semi-automatic evaluation without tangent does not require AD. Material law evaluations with adaptive step size and tangent are quite expensive. This is attributed to the effort of solving an ODE coupled with one for the derivative components. Again, the GPU improves the performance significantly, especially for evaluations with tangents. While there are no significant performance differences without tangent, ode23 is fastest with tangent, and is also competitive to the implicit Euler scheme on the CPU. Also, semi-automatic evaluation improves performance mostly on the CPU, and automatic evaluation has insignificant performance drawbacks on the GPU. The bad tangent performance of ode23s is related to register usage; this is explained in Section \ref{section:layout_profiling_limits}.

\begin{table*}
\centering
\begin{tabular}{c c c c c@{}}
\toprule
\textbf{architecture} & \textbf{ODE solver} & \textbf{evaluation strategy} & \makecell{\textbf{material law [s]}\\no tangent} & \makecell{\textbf{material law [s]}\\tangent}\tabularnewline
\midrule
CPU & ode12 & automatic & 2062.8 & 9070.5 \tabularnewline
CPU & ode12 & semi-automatic & 746.5 & 1844.7 \tabularnewline
CPU & ode23 & automatic & 1330.5 & 1683.9 \tabularnewline
CPU & ode23 & semi-automatic & 556.9 & 395.7 \tabularnewline
CPU & ode23s & semi-automatic & 1472.8 & 2885.7 \tabularnewline
\midrule
GPU & ode12 & automatic & 237.3 & 133.6 \tabularnewline
GPU & ode12 & semi-automatic & 247.4 & 121.3 \tabularnewline
GPU & ode23 & automatic & 238.1 & 73.2 \tabularnewline
GPU & ode23 & semi-automatic & 229.1 & 63.3 \tabularnewline
GPU & ode23s & semi-automatic & 235.1 & 508.5 \tabularnewline
\bottomrule
\end{tabular}
\caption{Total time spent on both types of material law evaluations with Runge-Kutta and Rosenbrock schemes. Compare also Table \ref{table:implicit_euler_timings}.}
\label{table:advanced_integrators_timings}
\end{table*}
\begin{table}
\centering
\begin{tabular}{c c@{}}
\toprule
\textbf{ODE solver} & \textbf{number of iterations} \tabularnewline
impl. Euler & 1681 \tabularnewline
ode12 & 1571 \tabularnewline
ode23 & 1527 \tabularnewline
ode23s & 1554 \tabularnewline
\bottomrule
\end{tabular}
\caption{Impact of ODE solver choice on number of elasticity solver iterations. The distinction CPU/GPU and the evaluation strategy types have no influence in this regard.}
\label{table:integrators_iterations}
\end{table}

\begin{figure*}[!ht]
	\centering
	\includegraphics{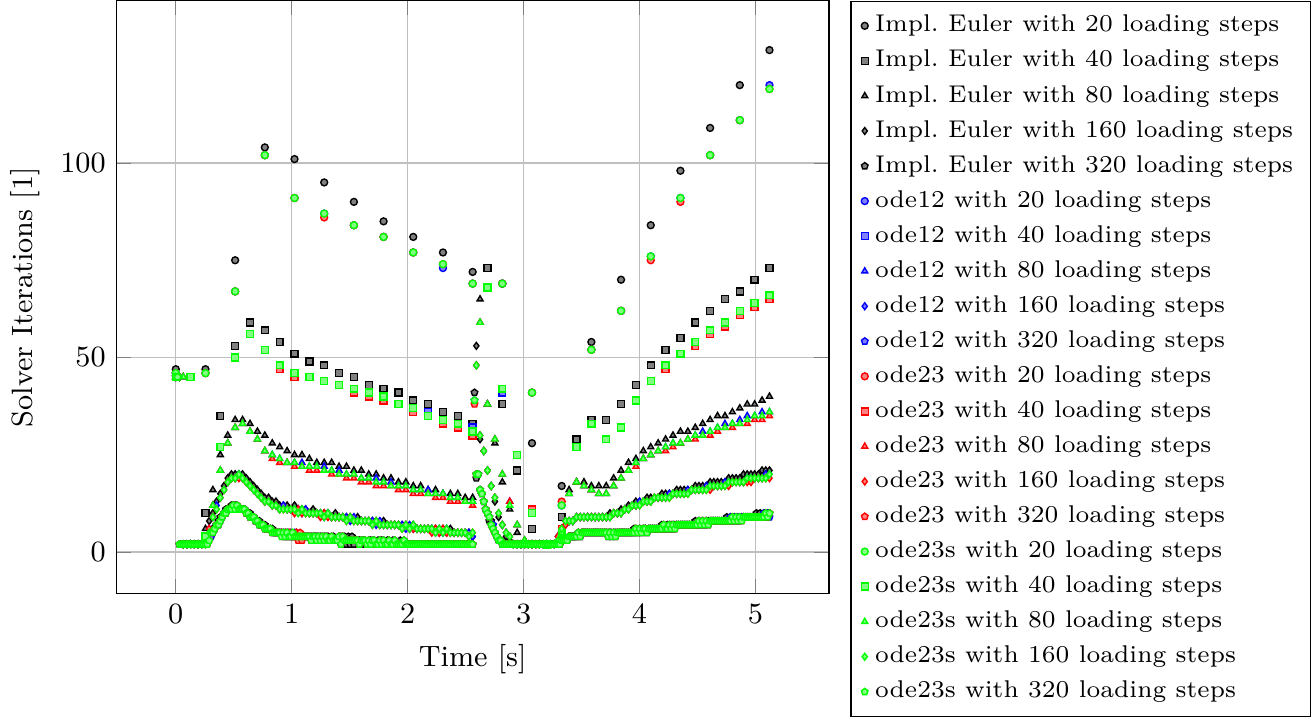}
    \caption{Iterations of FFT-based homogenization per loading step for the example of Section \ref{section:example}.}
    \label{figure:Results_Michel_Suquet_Composite_SolverIterations}
\end{figure*}

As can be seen in Table \ref{table:integrators_iterations} for the case of 80 loading steps, adaptive substeps tend to reduce the overall number of elasticity solver iterations so that there are less material law evaluations without tangent in total. Figure \ref{figure:Results_Michel_Suquet_Composite_SolverIterations}, however, reveals that the loading step size remains --- consistently across all ODE solvers --- the key influence factor on the number of iterations per loading step.

\begin{figure*}[!ht]
	\centering
    \includegraphics{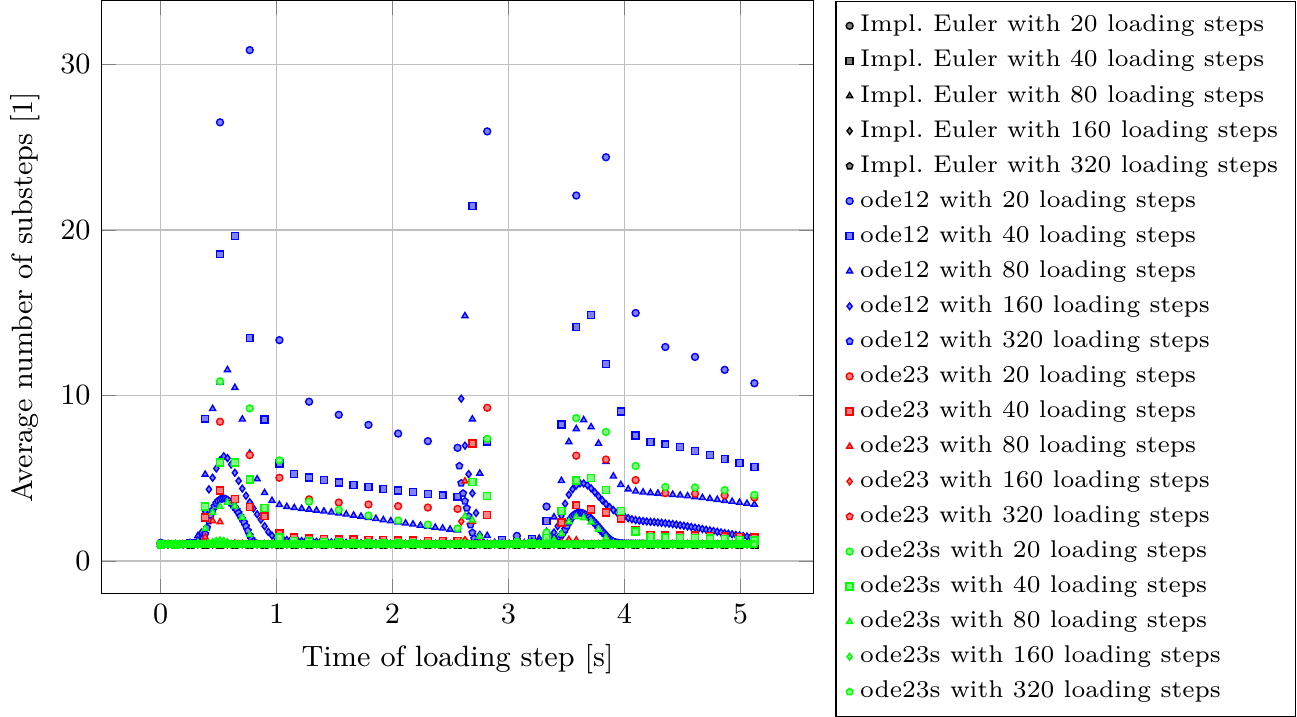}
    \caption{Spatially averaged number of ODE solver substeps per loading step for the example of Section \ref{section:example}. For each loading step, the number of substeps is plotted against the loading step's time.}
    \label{figure:Results_Michel_Suquet_Composite_substeps}
\end{figure*}

In Figure \ref{figure:Results_Michel_Suquet_Composite_substeps}, the average number of substeps per loading steps are visualized for the four different ODE solvers. By design, implicit Euler always uses one substep per loading step. For the other three solvers, the average number of substeps varies. It is strongly increasing when nonlinear effects occur in the composite. As expected, the first/second order solver ode12 needs the most substeps to reach the prescribed accuracy. The second/third order solvers ode23 and ode23s need a comparable number of substeps. Consequently, the semi-implicit and computationally more expensive ode23s cannot outperform the explicit ode23.

\subsection{Solution Accuracy}
\label{subsection:solution_accuracy}

FFT-based homogenization of Moulinec-Suquet \cite{Moulinec1998} applied to materials with nonlinear behaviour is subject to a spatial discretization error of the partial differential equation $\Div \sigma = 0$
investigated in detail by Schneider \cite{Schneider2015} and furthermore two types of time discretization errors. First, the interaction between different regions of the material (quadrature points) over time is neglected on the material law evaluation level. Second, each integration of the ordinary differential equations \eqref{eq:gsm_internal}, that is, each material law evaluation, introduces a local error in the internal variables.

For our example presented in Section \ref{section:example}, we study the influence of the adaptive time step size control on the overall error by comparing the ODE solvers presented above.

The stress response in loading direction is shown in Figure \ref{figure:Results_Michel_Suquet_Composite_sigma_xx}. As expected, due to the error control, all ODE solvers with adaptive time steps predict the same effective stress response within the given tolerances. Moreover, as can be seen in Figure \ref{figure:Error_Michel_Suquet_Composite_sigma_xx}, the error for coarse loading steps is reduced to approximately 30\% of the error of the implicit Euler solver. Thus, the error of the material law evolution, that is, the accuracy of the ODE solver, is dominating the overall error of the FFT-based based homogenization for this example.

\begin{figure*}[!ht]
	\centering
	\includegraphics{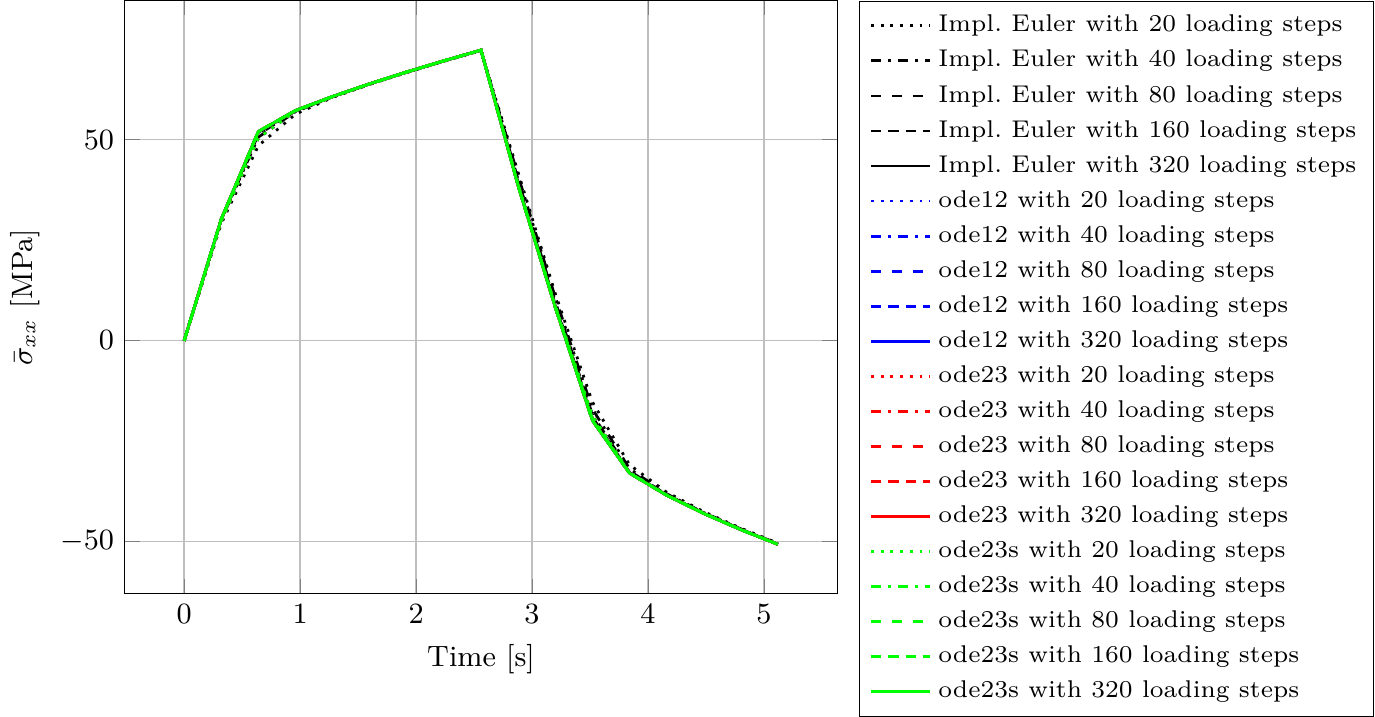}
    \caption{$\bar{\sigma}_{xx}$ for the example of Section \ref{section:example}.}
    \label{figure:Results_Michel_Suquet_Composite_sigma_xx}
\end{figure*}

\begin{figure*}[!ht]
    \centering
    \includegraphics{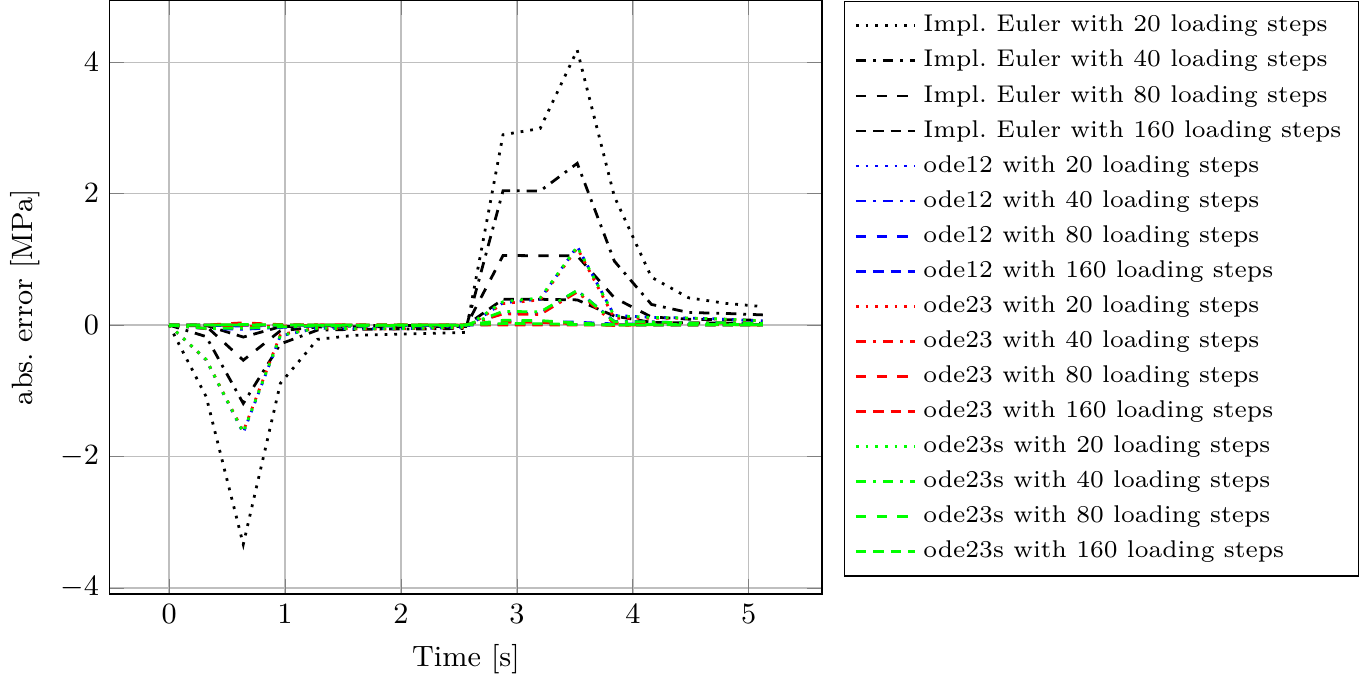}
    \caption{Difference of $\bar{\sigma}_{xx}$ to finest time discretization (320 loading steps) for the different ODE solvers.}
    \label{figure:Error_Michel_Suquet_Composite_sigma_xx}
\end{figure*}

For the tangential stiffness shown in Figures \ref{figure:Results_Michel_Suquet_Composite_C11} and \ref{figure:Results_Michel_Suquet_Composite_C12}, the results depend on the time discretization as explained in detail in Section \ref{subsection:rosenbrock_runge_kutta}. Therefore, we cannot perform a convergence analysis with respect to the loading step size. We observe that all ODE solvers with adaptive time step size control predict almost the same algorithmic tangent due to the error control. The differences observed between single implicit Euler steps and schemes with adaptive substeps are in accordance with the example on the relative error amplification in Section \ref{section:automatic_evaluation}. Note that the tangent formula \eqref{eq:tangent_formula} reads for the potentials \eqref{eq:michel_suquet_omega} and \eqref{eq:michel_suquet_psi}
\begin{equation*}
	C_{n+1}=C^\elastic\of{I-\frac{\intd \viscop[n+1]{\varepsilon}}{\intd\varepsilon_{n+1}}},
\end{equation*}
that is, linear combinations of errors as depicted in Figure \ref{fig:relative_error_tangent} are substracted from the components of the elastic stiffness matrix. This effect regards voxels that follow the Michel Suquet law and can still be seen in the effective stiffness.

\begin{figure*}[!ht]
	\centering
    \includegraphics{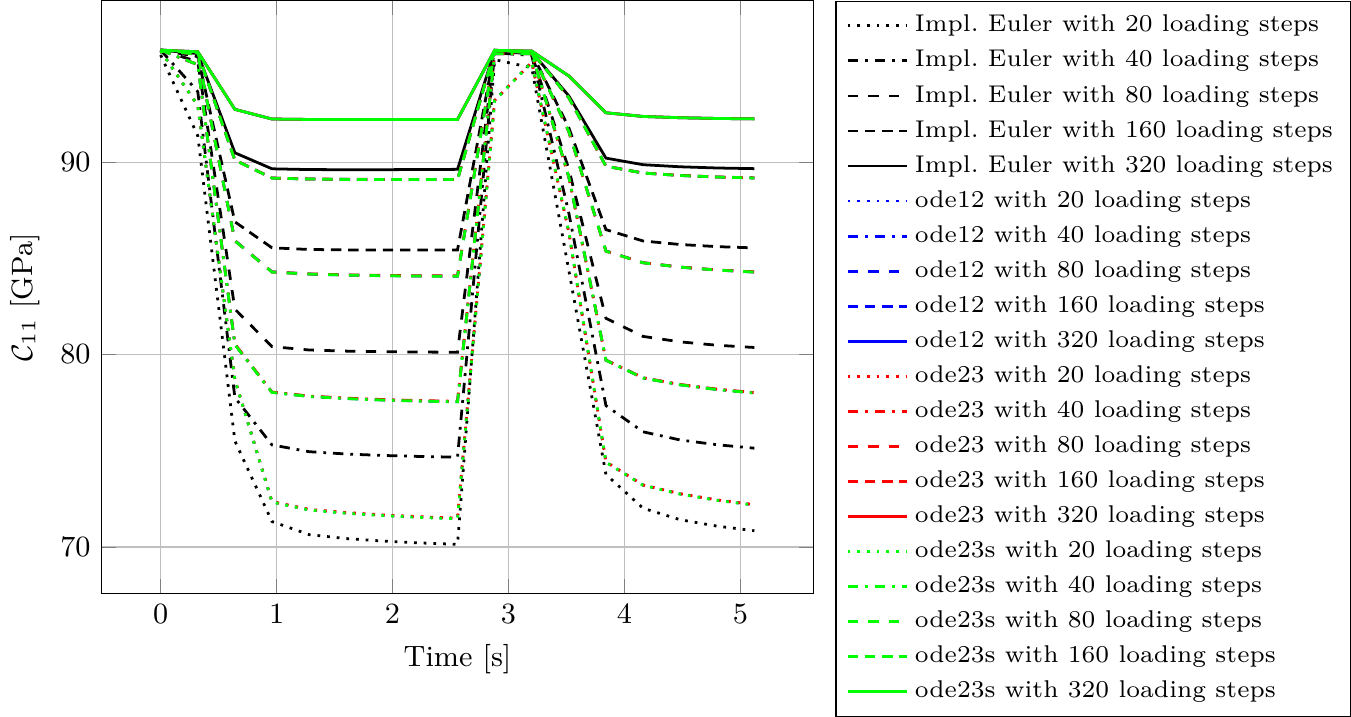}
    \caption{$\mathcal{C}_{11}$ for the example of Section \ref{section:example}.}
    \label{figure:Results_Michel_Suquet_Composite_C11}
\end{figure*}

\begin{figure*}[!ht]
	\centering
    \includegraphics{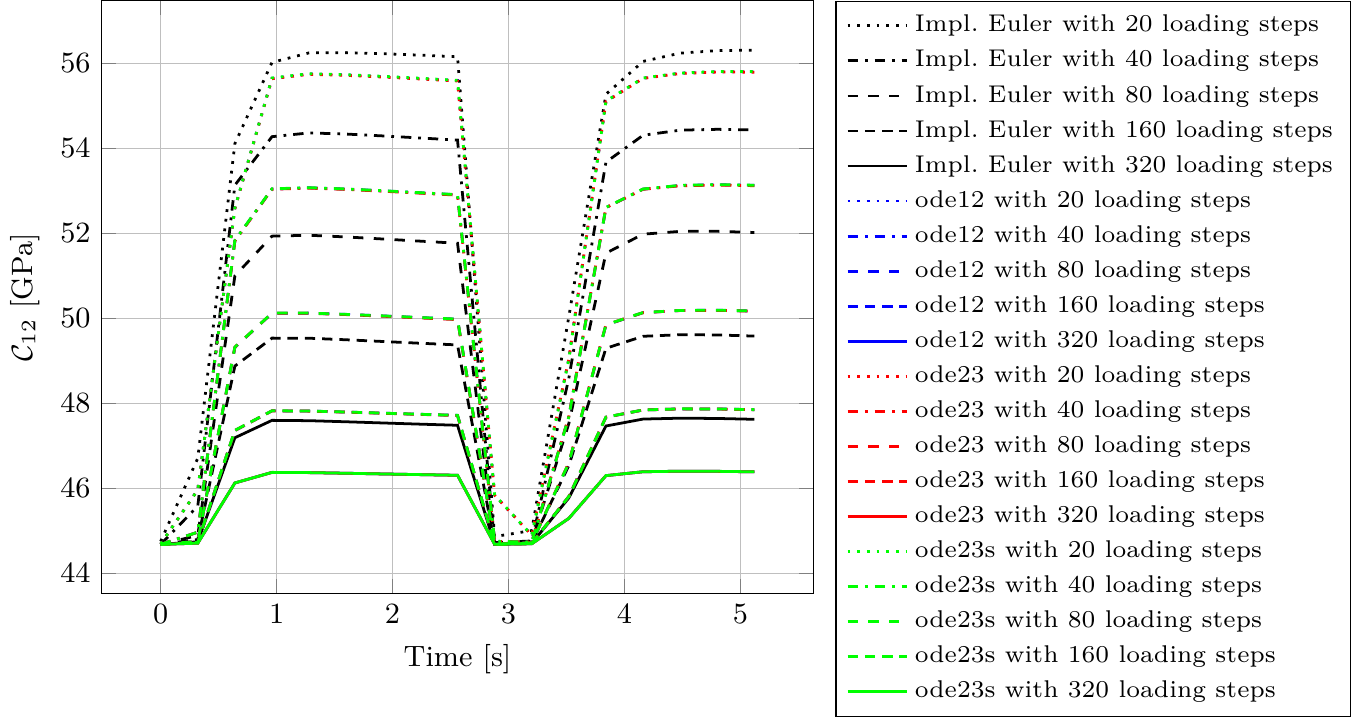}
    \caption{$\mathcal{C}_{12}$ for the example of Section \ref{section:example}.}
    \label{figure:Results_Michel_Suquet_Composite_C12}
\end{figure*}

\subsection{Stress-Driven Error Control}
\label{subsec:stress_driven_error_control}

Internal variables do not always have a physical meaning, and the material law outputs that are of immediate relevance to the elasticity solver are $\sigma_{n+1}$ and $C_{n+1}$. Its convergence test, for example, amounts to an equilibrium check of the stress field \cite{Moulinec1998}, and the material tangents are used to determine a linear elastic reference material \cite{Michel2001,Eisenlohr2013,Kabel2014}. In the material law evaluations, however, the tolerances specified for the ODE solver relate to an error in the internal variables. We control the error in $a_{n+1}$ and --- if we apply Corollary \ref{corollary:ad_integration_scheme} --- as well the error in $\frac{\mathrm{d}a_{n+1}}{\mathrm{d}\varepsilon_{n+1}}$.

In the GSM given by Equations \eqref{eq:michel_suquet_omega} and \eqref{eq:michel_suquet_psi}, for example, the stress relationship \eqref{eq:gsm_sigma} turns into
$\sigma=C^\elastic (\varepsilon - \viscop\varepsilon)$,
that is, any error in $\viscop\varepsilon$ enters $\sigma$ multiplied by the elastic stiffness tensor. Depending on the specific instance of $C^\elastic$, it might be necessary to adapt the tolerances of the ODE solver to end up with stress values that are sufficiently accurate for the PDE solver. This is avoided by an error control on the ODE level that is directly tied to the accuracy of the stresses.

The step size control mechanism from \cite{Hairer1993} captures the deviation between two ODE solutions of different order of convergence in an error measure. Depending on the error, steps are accepted or rejected and the step size is adapted accordingly. Instead of using the internal variable approximations directly in the error measure, we transform them together with the adequate linear interpolation between $\varepsilon_n$ and $\varepsilon_{n+1}$ for the substep of interest via the relationship \eqref{eq:gsm_sigma} into a pair of stresses. If $\sigma$ depends --- as above --- linearly or, more generally, Lipschitz on the internal variables, this yields a pair of stresses with the analogous order relations. The rationale of the step size control carries over, and we evaluate the error measure on the stresses instead. If we solve additionally for the derivative $\frac{\mathrm{d}a_{n+1}}{\mathrm{d}\varepsilon_{n+1}}$, the same evaluation of \eqref{eq:gsm_sigma} (performed on forward AD types instead) transforms additionally the approximations of the internal variable derivatives into a corresponding pair of material tangents that may then enter the error measure in the same way the derivative components did before. This way, we control the error in $\sigma_{n+1}$ and $C_{n+1}$.

As can be seen in Figure \ref{figure:Error_Michel_Suquet_Composite_sigma_xx_ErrorMeasure}, stress-driven error control also reduces the impact of the ODE solver choice on the effective stress response for all numbers loading steps.

\begin{figure*}[!ht]
	\centering
    \includegraphics{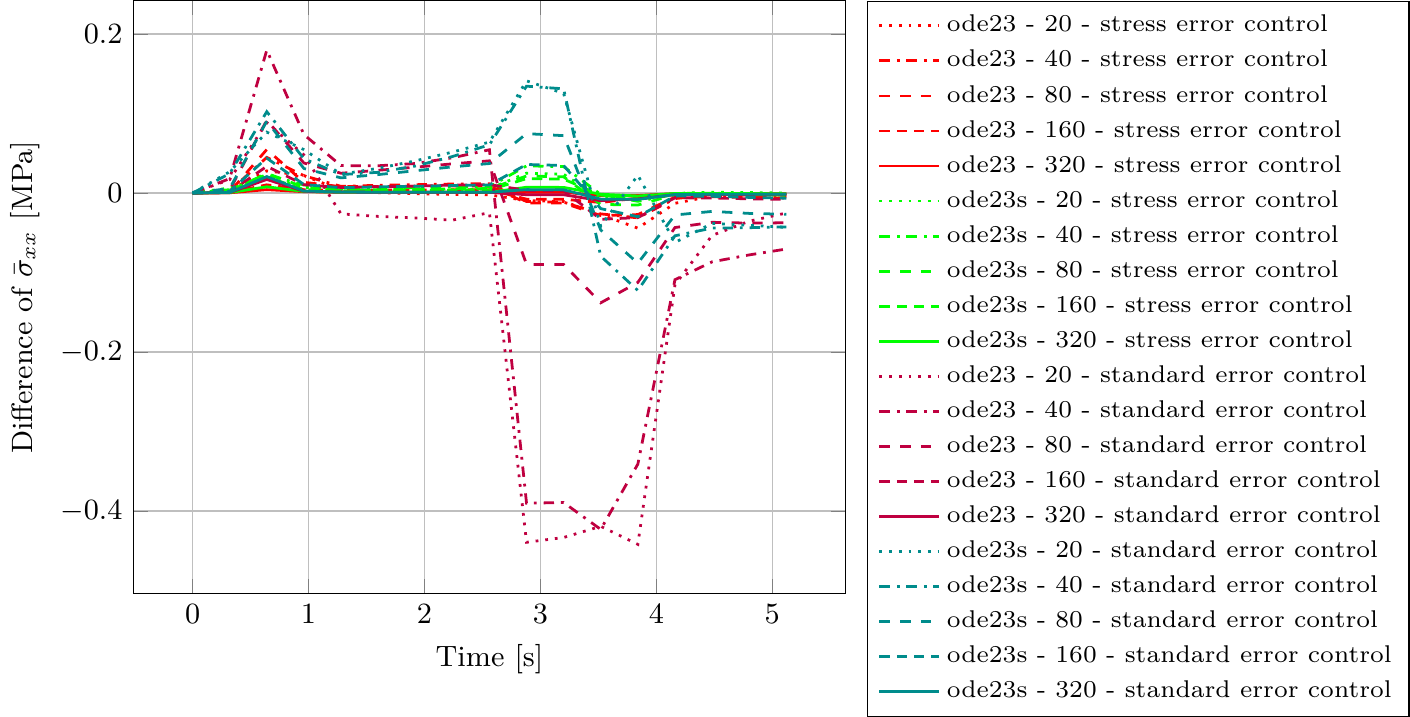}
    \caption{Difference of $\bar{\sigma}_{xx}$ for ode23 and ode23s compared to ode12 at the same loading path discretization and with the same error measure for step size control.}
    \label{figure:Error_Michel_Suquet_Composite_sigma_xx_ErrorMeasure}
\end{figure*}

Similar ideas can be employed for the convergence criterion of Newton's method in the schemes from Section \ref{subsec:single_implicit_euler_step}. Instead of iterating until convergence in $a$, we may compute the stress resulting from the current iterate via \eqref{eq:gsm_sigma} in each Newton iteration and converge $\sigma$ instead.

\section{Automatic Differentiation on GPUs}

\label{section:automatic_differentiation_on_gpus}

To summarize the basic ideas of automatic differentiation, we view a floating point computation with fully evaluated control flow as a function $x\mapsto y$ that is composed of elementary mathematical operations like $+$, $\cdot$ or standard math library functions like $\sin$. If we differentiate the composed operations according to the chain rule, we obtain the \emph{algorithmic derivative} of the computer program. \emph{Automatic differentiation} deals with techniques that obtain algorithmic derivatives in an automatic fashion. A comprehensive introduction is given in \cite{Griewank2008}.

As both $\omega$ and $\Psi$ are scalar valued functions and have --- with respect to both $\varepsilon$ and $a$ --- more inputs than outputs, it seems appropriate to use the \emph{reverse mode of automatic differentiation} to evaluate the partial derivatives on the right hand sides of the GSM constitutive equations \eqref{eq:gsm_sigma} and \eqref{eq:gsm_internal}. $C_{n+1}$, on the other hand, arises as the derivative of $\sigma_{n+1}$ with respect to $\varepsilon_{n+1}$, that is, six Voigt components with respect to six Voigt components. We compute it with the \emph{forward mode of automatic differentiation}, possibly the \emph{forward vector mode}. To compute both the partials and $C_{n+1}$ with AD at the same time, we combine the forward and reverse mode in an adjoints of tangents fashion \cite{Griewank2008}. While the computation is generally executed on a forward AD data type, all local evaluations of partials are obtained by additional applications of the reverse mode. In the context of semi-automatic ode23s, we use the second order forward (vector) mode for the Jacobians and tangents.

The implementation of the first and second order forward (vector) mode follows the same principles as CPU implementations like \cite{Sagebaum2019}. The reverse mode of AD, however, is subject to a global information problem that is typically solved by \emph{taping}. The sequence of operations is first executed in forward direction and remembered together with all intermediate results. Then, the corresponding sequence of derivatives is evaluated according to the chain rule in reverse order. On the GPU, this memory-intensive approach is prohibitive. Since the reverse mode of AD is only needed in a very local manner, we may replace taping by recomputations: If an intermediate value is required during reverse evaluation, the sequence of operations is partially re-evaluated in forward direction up to the required point. Similar approaches are pursued in \cite{Leppkes2017}.

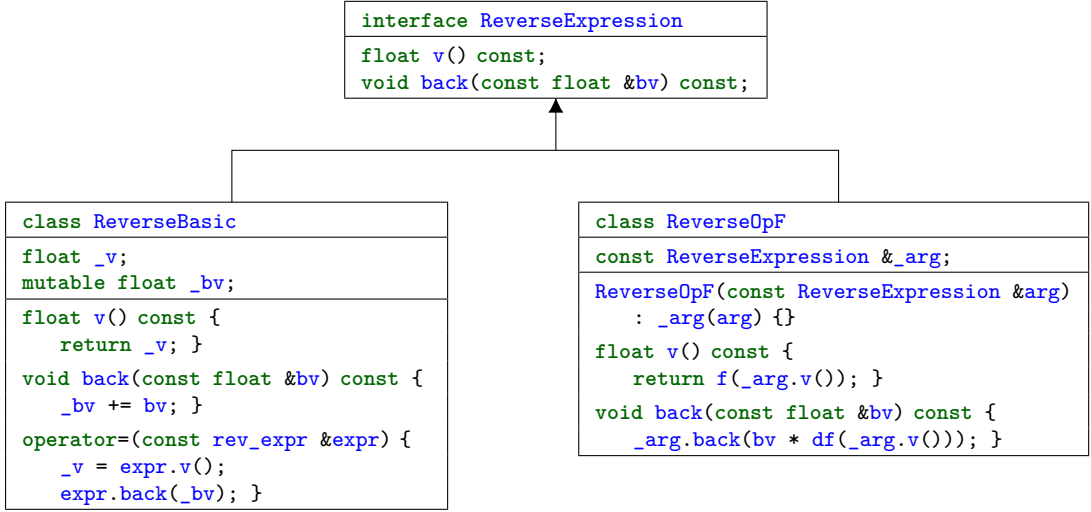
\begin{figure*}
\begin{center}
\begin{tabular}{|l|}
	\hline
	\makecell[l]{\lstinline{interface ReverseExpression}} \Tstrut\Bstrut\\\hline
	\makecell[l]{\lstinline{float v()} \lstinline{const;}} \Tstrut\\
	\makecell[l]{\lstinline{void back(const float &bv)} \lstinline{const;}}\Bstrut\\\hline
\end{tabular}\strut\par
\tikz[remember picture, overlay] \node [coordinate] (rev_expr_bottom) at (0,0.05\arraystretch) {};
\end{center}\strut\par
\vspace{0.5cm}
\begin{minipage}[t]{0.42\textwidth}
\tikz[remember picture, overlay] \coordinate (rev_basic_top) at (0.5\textwidth,0.25\arraystretch);
\begin{tabular}[t]{|l|}
	\firsthline
	\makecell[l]{\lstinline{class ReverseBasic}} \Tstrut\Bstrut\\\hline
	\makecell[l]{\lstinline{float _v;}} \Tstrut\\
	\makecell[l]{\lstinline{mutable float _bv;}} \Bstrut\\\hline
	\makecell[l]{\lstinline{float v()} \lstinline{const \{}} \Tstrut\\
	\hspace{0.5cm}\makecell[l]{\lstinline{return _v; \}}}\Bstrut\\
	\makecell[l]{\lstinline{void back(const float &bv)} \lstinline{const \{}} \Tstrut\\
	\hspace{0.5cm}\makecell[l]{\lstinline{_bv += bv; \}}}\Bstrut\\
	\makecell[l]{\lstinline{operator=(const rev_expr &expr)} \lstinline{\{}}\Tstrut\\
	\hspace{0.5cm}\makecell[l]{\lstinline{_v = expr.v();}}\\
	\hspace{0.5cm}\makecell[l]{\lstinline{expr.back(_bv); \}}}\Bstrut\\\hline
\end{tabular}
\end{minipage}
\hfill
\begin{minipage}[t]{0.48\textwidth}
\tikz[remember picture, overlay] \coordinate (rev_op_f_top) at (0.5\textwidth,0.25\arraystretch);
\begin{tabular}[t]{|l|}
	\firsthline
	\makecell[l]{\lstinline{class ReverseOpF}}\Tstrut\Bstrut\\\hline
	\makecell[l]{\lstinline{const ReverseExpression &_arg;}}\Tstrut\Bstrut\\\hline
	\makecell[l]{\lstinline{ReverseOpF(const ReverseExpression &arg)}}\Tstrut\\
	\hspace{0.5cm}\makecell[l]{\lstinline{: _arg(arg)} \lstinline{\{\}}}\Bstrut\\
	\makecell[l]{\lstinline{float v()} \lstinline{const \{}} \Tstrut\\
	\hspace{0.5cm}\makecell[l]{\lstinline{return f(_arg.v()); \}}}\Bstrut\\
	\makecell[l]{\lstinline{void back(const float &bv)} \lstinline{const \{}} \Tstrut\\
	\hspace{0.5cm}\makecell[l]{\lstinline{_arg.back(bv * df(_arg.v())); \}}}\Bstrut\\\hline
\end{tabular}\strut
\end{minipage}
\begin{tikzpicture}[overlay, remember picture]
	\node [below=0.7cm of rev_expr_bottom, coordinate] (helper) {};
	\draw (rev_basic_top) |- (helper);
	\draw (rev_op_f_top) |- (helper);
	\draw [-{Latex[length=2mm,width=2mm]}] (helper) -- (rev_expr_bottom);
\end{tikzpicture}
\caption{Schematic implementation of the reverse mode of AD on the expression level. \lstinline{_arg.back(bv * df(_arg.v()))} is the classical backpropagation formula \protect\cite{Griewank2008}.}
\label{figure:reverse_mode}
\end{figure*}

This can be realized by an \emph{operator overloading} ansatz at low computational overhead on the expression level. We employ \emph{expression template techniques} that have previously been shown to perform well for the treatment of right hand sides in the forward mode of AD \cite{Phipps2012} and in Jacobi taping \cite{Hogan2014} as well as primal value taping \cite{Sagebaum2018} in the reverse mode of AD. Here, we use expression templates to convert a composite operation into a structured data type that represents the computational graph and allows for its traversal in forward and reverse direction. This way, the structure of the computation is fully exposed to the compiler and can be optimized during compilation. The \emph{curiously recurring template pattern} is used to shift overhead due to the interface in the inheritance tree in Figure \ref{figure:reverse_mode} from run time to compile time.

Figure \ref{figure:reverse_mode} showcases the reverse mode without additional tangents using the example of a unary elementary operation \lstinline{f()}. The interface \lstinline{ReverseExpression} defines a routine \lstinline{v()} for forward evaluation and a routine \lstinline{back()} for backpropagation of derivatives. On the one hand, it is implemented as a type \lstinline{ReverseBasic} that contains actual data, that is, a primal value \lstinline{_v} and an adjoint value \lstinline{_bv}. On the other hand, there are derived types that stand for applied elementary operations such as \mbox{\lstinline{ReverseOpF}.} They are created by operation overloads such as
\begin{center}
	\begin{tabular}{l}
		\makecell[l]{\lstinline{ReverseOpF f(const ReverseExpression &expr)}} \\
		\makecell[l]{\lstinline{\{} \lstinline{return ReverseOpF(expr);} \lstinline{\}}} \\
	\end{tabular}
\end{center}
that do not immediately apply \lstinline{f()} but store a reference to the arguments in the returned object. Types such as \lstinline{ReverseOpF} implement the interface in a way that allows for the forward and reverse evaluation of the computational graph. A call to \lstinline{v()} causes the forward evaluation of \lstinline{_arg} and subsequent application of \lstinline{f()}. A call to \lstinline{back()} propagates derivative values in reverse direction where \lstinline{df()} stands for the derivative of \lstinline{f()} and must be implemented explicitly. The call \lstinline{_arg.v()} in \mbox{\lstinline{ReverseOpF::back()}} causes forward re-evaluations. This extends analogously to $n$-ary operations and additional forward and reverse evaluation of tangents for second order derivatives. Consider a code segment
\begin{center}
	\begin{tabular}{l}
	\makecell[l]{\lstinline{// initialize primal components}}\\
	\makecell[l]{\lstinline{// set derivative values to 0}}\\
	\makecell[l]{\lstinline{ReverseBasic arg1 = ..., arg2 = ..., ...;}} \\\\
	\makecell[l]{\lstinline{ReverseBasic result;}} \\
	\makecell[l]{\lstinline[mathescape=true]|result._bv = 1.0; //$\hspace{3pt}$seeding|} \\
	\makecell[l]{\lstinline{result = CompositeExpression(}} \\
	\hspace{2cm}\makecell[l]{\lstinline{arg1, arg2, ...);}} \\
	\end{tabular}
\end{center}
where \lstinline{CompositeExpression} stands for a composition of multiple elementary operations. Each elementary operation must be implemented according to Figure \ref{figure:reverse_mode}. The operation overloads are used to build up the computational graph of this right hand side and in the course of the assignment to \lstinline{result}, \lstinline{ReverseBasic::operator=()} is used to trigger its forward and subsequent reverse evaluation. 
In the end, \lstinline{argn._bv} carries the machine accurate derivative of \lstinline{result._v} with respect to \lstinline{argn._v} where \lstinline{n = 1, 2, ...}.

The proposed AD tool can be implemented in \cpp using \cpp{}11 features that are supported both by standard compilers such as \verb|g++| and by Nvidia's CUDA compiler driver \verb|nvcc|. Particularly, the AD tool can be applied both inside OpenMP threads and CUDA kernels.

We improve the performance of the AD tool by some adaptions that are specific to our problem and setting.
\begin{enumerate}
	\item During expression tree forward traversal, it is possible to evaluate primal values only once and store them in the nodes of the tree \cite{Phipps2012}. However, to consume as little memory as possible, we use recomputations instead. This is especially important for the GPU on which memory operations are costly and the number of registers used per thread can limit parallel execution.
	\item Instead of a recursive ansatz for higher order derivatives, we implement second order expressions explicitly. This helps the compiler with the identification and elimination of common subexpressions, which it cannot always do automatically.
	\item In the computation of the partials of $\omega$, we are always only interested in the derivative with respect to either $\varepsilon$ or $a$ but never both. If we compute the derivative with respect to one, there is no need to propagate derivative values back to the other. Therefore, we provide mixed order expressions that actively avoid reverse propagation of derivative values to lower order type arguments.
\end{enumerate}

The AD tool can only differentiate single expressions in reverse order and is overall limited to first and second order derivatives. The first and second order forward (vector) mode, however, are not restricted to single expressions and can be applied to general codes, like the ODE solvers in the case of AutoMat. In the presented design, automatic differentiation takes exclusively place in GPU registers (sometimes spilled but mostly actual, see Table \ref{table:ptxas}).

\section{Computational Layout, Profiling and Performance Limiters}
\label{section:layout_profiling_limits}

The fields for the internal variables, the current strain field and the predicted strain field, that is, the material law inputs for all voxels, reside in host memory. In general, GPU memory is not large enough to hold all of them at the same time and the elasticity solver still runs on the CPU. Furthermore, data might reside in host memory in an array-of-struct layout that does not suite GPU computing and due to the heterogenity of the material, data for all voxels of a specific material law might be arranged in memory in a non-contiguous manner. Therefore, we divide the workload into multiple chunks of fixed size, in a way that GPU memory can at least hold the material law inputs and outputs of one or few chunks. In host memory, we allocate at least one staging area of chunk size and page-locked type that allows for fast CPU-GPU data exchange. On the host side, we copy the material law inputs of a chunk into the staging area in an OpenMP parallel manner. In doing so, we arrange them in a contiguous manner in a struct-of-arrays layout, and might convert from double to single precision. Then, we process the staging area with multiple CUDA streams. Each stream copies part of the inputs to the GPU and issues the corresponding material law evaluations. We use one CUDA thread per material law evaluation and a small multiple of 32 as block size for the computational grid. Once the evaluations are done, the stream copies the material law outputs back to the staging area. The purpose of multiple streams is an overlap of CPU-GPU data exchange with GPU computations. Once the entire staging area is processed, the material law outputs are collected from the staging area, transformed back to the original layout and otherwise postprocessed as required by the elasticity solver in an OpenMP parallel manner. By means of multiple staging areas, an overlap of CPU and GPU workloads can be achieved: During GPU computations, transformations of inputs and outputs involving other staging areas can already take place on the host side. We observed no benefits for more than two staging areas.

The CPU-GPU overlap becomes evident in Table \ref{table:cpu_gpu_overlap_timings}. For material law evaluations without tangent, the time spent on material law evaluation is determined by the time it takes to stage and collect the data. CPU-GPU data exchange and GPU computations overlap almost completely with the CPU workloads. The minimum time needed for exchange of the combined data over the PCI Express bus (assuming full bandwidth and perfect overlap of both transfer directions) gives an impression of the amount of time that is at least hidden behind CPU workloads.
\begin{table*}
\centering
\begin{tabular}{c c c c c c c@{}}
\toprule
\makecell{\textbf{ODE}\\\textbf{solver}} & \makecell{\textbf{eval.}\\\textbf{strategy}} & \textbf{tangent} & \makecell{\textbf{staging}\\\textbf{[s]}} & \makecell{\textbf{wait for}\\\textbf{GPU [s]}} & \makecell{\textbf{collecting}\\\textbf{[s]}} & \makecell{\textbf{PCIe}\\\textbf{bound [s]}} \tabularnewline
\midrule
impl.~Euler & conventional & no & 133.4 & 0.6 & 110.8 & 42.7\tabularnewline
impl.~Euler & semi-automatic & no & 136.2 & 1.0 & 115.4 & 42.7\tabularnewline
impl.~Euler & automatic & no & 136.0 & 0.7 & 116.0 & 42.7\tabularnewline
ode12 & automatic & no & 124.8 & 0.6 & 107.2 & 39.9\tabularnewline
ode12 & semi-automatic & no & 129.6 & 0.6 & 112.3 & 39.9\tabularnewline
ode23 & automatic & no & 124.7 & 0.6 & 108.0 & 38.8\tabularnewline
ode23 & semi-automatic & no & 122.4 & 0.6 & 101.6 & 38.8\tabularnewline
ode23s & semi-automatic & no & 123.9 & 2.1 & 104.5 & 39.5\tabularnewline
\midrule
impl.~Euler & conventional & yes & 6.6 & 6.8 & 21.4 & 2.4\tabularnewline
impl.~Euler & semi-automatic & yes & 6.8 & 7.7 & 23.7 & 2.4\tabularnewline
impl.~Euler & automatic & yes & 6.7 & 7.6 & 24.0 & 2.4\tabularnewline
ode12 & automatic & yes & 6.6 & 97.3 & 29.7 & 2.4\tabularnewline
ode12 & semi-automatic & yes & 6.9 & 79.6 & 34.7 & 2.4\tabularnewline
ode23 & automatic & yes & 6.7 & 30.9 & 35.6 & 2.4\tabularnewline
ode23 & semi-automatic & yes & 6.6 & 26.8 & 29.9 & 2.4\tabularnewline
ode23s & semi-automatic & yes & 6.7 & 467.6 & 34.3 & 2.4\tabularnewline
\bottomrule
\end{tabular}
\caption{Refinement of timings for GPU configurations from Tables \ref{table:implicit_euler_timings} and \ref{table:advanced_integrators_timings} into CPU workloads and non-overlapped GPU workloads. Includes also lower time bound for the PCIe data exchange. Some variations between material law evaluations without tangent are due to differences in the number of elasticity solver iterations. The variations in the PCIe bounds indicate this extent.}
\label{table:cpu_gpu_overlap_timings}
\end{table*}
The exemplary profilings presented in Figure \ref{figure:profilings_no_tangent} show that the GPU compute time is in turn dominated by CPU-GPU data exchange, and due to overlap mostly hidden behind it.
\begin{figure*}[ht]
	\centering
	\begin{tabular}{cc}
		(1) & \begin{minipage}{0.9\textwidth}\includegraphics[width=\textwidth]{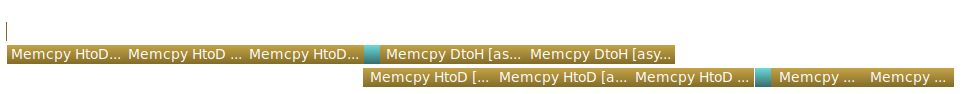}\end{minipage} \\
		(2) & \begin{minipage}{0.9\textwidth}\includegraphics[width=\textwidth]{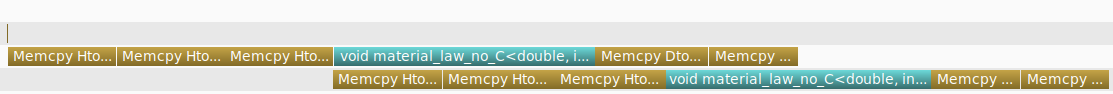}\end{minipage} \\
		(3) & \begin{minipage}{0.9\textwidth}\includegraphics[width=\textwidth]{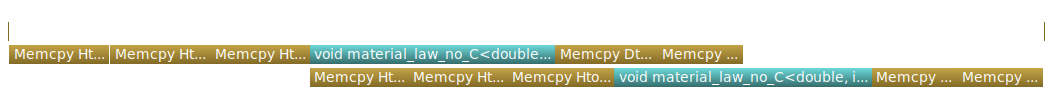}\end{minipage} \\
		(4) & \begin{minipage}{0.9\textwidth}\includegraphics[width=\textwidth]{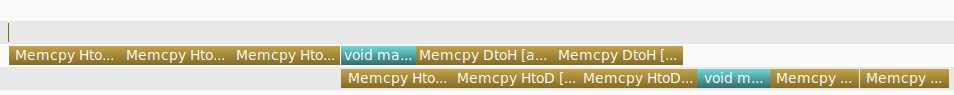}\end{minipage}
	\end{tabular}
	\caption{Profilings performed for material law evaluations without tangent and implicit Euler --- conventional (1), semi-automatic (2), automatic (3) --- and automatic ode23 (4). Indicates overlap and relative duration within configurations, time scales vary between (1)--(4). Staging area processed with two CUDA streams. CPU-GPU data exchange brown, computations blue.\protect\footnotemark}
	\label{figure:profilings_no_tangent}
\end{figure*}
\begin{figure*}[ht]
	\centering
	\begin{tabular}{cc}
		(1) & \begin{minipage}{0.9\textwidth}\includegraphics[width=\textwidth]{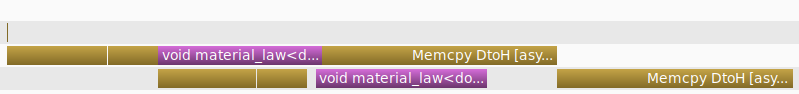}\end{minipage} \\
		(2) & \begin{minipage}{0.9\textwidth}\includegraphics[width=\textwidth]{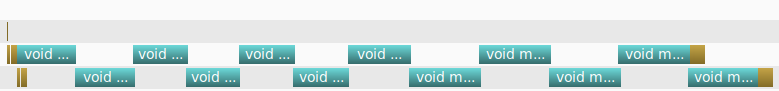}\end{minipage} \\
		(3) & \begin{minipage}{0.9\textwidth}\includegraphics[width=\textwidth]{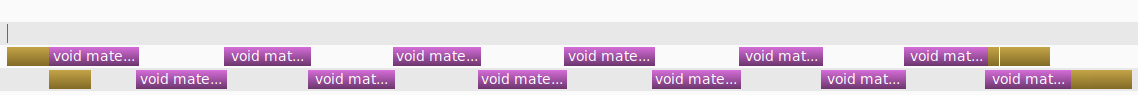}\end{minipage} \\
		(4) & \begin{minipage}{0.9\textwidth}\includegraphics[width=\textwidth]{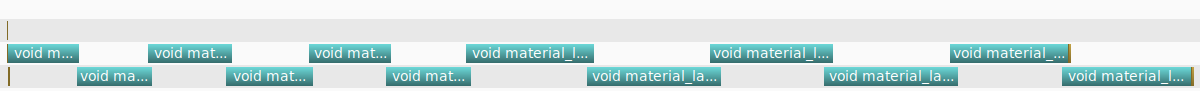}\end{minipage}
	\end{tabular}
	\caption{Profilings performed for automatic material law evaluations with tangent --- implicit Euler (1), ode12 (2), ode23 (3) --- and semi-automatic ode23s (4). Time scales vary between (1)--(4). Note the six evaluation steps between data transfer in (2)--(4). Staging area processed with two CUDA streams. CPU-GPU data exchange brown, computations blue/purple.\textsuperscript{\ref{footnote:nvvp}}}
	\label{figure:profilings_with_tangent}
\end{figure*}

For material law evaluations with tangent, the observations are different. Here, staging and collecting cannot hide all GPU workloads, in particular the GPU computations which are also more expensive than the CPU-GPU data exchange. This has two reasons. First, the the postprocessing step for the tangent or solving the coupled ODE system, respectively, is in itself computationally more expensive. The derivative components, however, also increase the memory footprint of the GPU kernels, in particular the number of registers used per thread. This can be seen in Table \ref{table:ptxas}.
\begin{table*}
\centering
\begin{tabular}{c c c c c c c@{}}
\toprule
\makecell{\textbf{ODE}\\\textbf{solver}} & \makecell{\textbf{eval.}\\\textbf{strategy}} & \textbf{tangent} & \makecell{\textbf{stack frame}\\\textbf{[bytes]}} & \makecell{\textbf{spill stores}\\\textbf{[bytes]}} & \makecell{\textbf{spill loads}\\\textbf{[bytes]}} & \makecell{\textbf{regis-}\\\textbf{ters}} \tabularnewline
\midrule
impl.~Euler & conventional & no & 0 & 0 & 0 & 73 \tabularnewline
impl.~Euler & semi-automatic & no & 656 & 0 & 0 & 184 \tabularnewline
impl.~Euler & automatic & no & 656 & 0 & 0 & 186 \tabularnewline
ode12 & automatic & no & 0 & 0 & 0 & 156 \tabularnewline
ode12 & semi-automatic & no & 0 & 0 & 0 & 136 \tabularnewline
ode23 & automatic & no & 0 & 0 & 0 & 206 \tabularnewline
ode23 & semi-automatic & no & 0 & 0 & 0 & 174 \tabularnewline
ode23s & semi-automatic & no & 944 & 224 & 392 & 255 \tabularnewline
\midrule
impl.~Euler & conventional & yes & 0 & 0 & 0 & 114 \tabularnewline
impl.~Euler & semi-automatic & yes & 960 & 0 & 0 & 173 \tabularnewline
impl.~Euler & automatic & yes & 960 & 0 & 0 & 198 \tabularnewline
ode12 & automatic & yes & 0 & 0 & 0 & 246 \tabularnewline
ode12 & semi-automatic & yes & 0 & 0 & 0 & 215\tabularnewline
ode23 & automatic & yes & 368 & 320 & 424 & 255\tabularnewline
ode23 & semi-automatic & yes & 168 & 88 & 136 & 255 \tabularnewline
ode23s & semi-automatic & yes & 4416 & 4196 & 3952 & 255 \tabularnewline
\bottomrule
\end{tabular}
\caption{\texttt{ptxas} info for configurations from Table \ref{table:cpu_gpu_overlap_timings} (double precision). Indicates resources consumed per CUDA thread.}
\label{table:ptxas}
\end{table*}
This limits the overall number of threads that can run in parallel, and it is important to keep that number small. To that end, all ODE solvers with adaptive step size among the GPU configurations with tangent are subject to another performance optimization. Instead of propagating all six tangent directions simultaneously through one material law evaluation with the forward vector mode, we re-evaluate each material law six times, each with the standard forward mode and one tangent direction, i.\,e.~we compute $C_{n+1}$ column by column.
This can also be seen in Figure \ref{figure:profilings_with_tangent}. Note that this has no influence on staging, collecting or the amount of CPU-GPU data exchange. We trade memory for computations on the GPU, and the performance benefits of kernels with smaller memory footprint outweigh the additional effort incurred by the re-evaluations.

Interestingly, the CPU-GPU data exchange is --- due to overlap and the cost of staging and collecting --- in none of the configurations discussed above a key limiting factor. Nonetheless, our implementation of the material law from Section \ref{section:example} reduces that data. Material law evaluations with tangent, for example, copy back $C_{n+1}$ but neither stresses nor internal variables. Specifically for the GSM given by \eqref{eq:michel_suquet_omega}, \eqref{eq:michel_suquet_psi}, we exploit $\viscop\varepsilon\in\operatorname{range}(\dev)$, i.\,e.~one component of the viscoplastic strain can be eliminated and is computed on the fly in the implementations of $\omega$ and $\Psi$ from the others instead.
\footnotetext{Generated with Nvidia Visual Profiler, \url{https://developer.nvidia.com/nvidia-visual-profiler}.\label{footnote:nvvp}}

The effect of using AutoMat on the total runtime of FFT-based homogenization is summarized in Table \ref{table:total_runtimes}. On the CPU, ode23 is the best choice. It needs approximately the same time as our conventional implementation and gives more precise results according to Section \ref{subsection:solution_accuracy}. On the GPU, the choice of the ODE solver does not influence the total runtime significantly with the notable exception of ode23s. For all other ODE solvers, AutoMat accelerates the FFT-based homogenization method by a factor of more than two on the GPU. For ode23s, this holds only true if the reference material is not updated. In all other cases, the ODE solver can be chosen without performance considerations on the GPU.
\begin{table*}
\centering
\begin{tabular}{c c c c c@{}}
\toprule
\textbf{architecture} & \textbf{ODE solver} & \textbf{evaluation strategy} & \makecell{\textbf{time [s]}\\with update} & \makecell{\textbf{time [s]}\\without update}\tabularnewline
\midrule
CPU & implicit Euler & conventional & 1162.74 & 1133.61\tabularnewline
CPU & implicit Euler & automatic & 6282.41 & 6055.91\tabularnewline
CPU & implicit Euler & semi-automatic & 2185.59 & 2100.79\tabularnewline
CPU & ode12 & automatic & 11555.43 & 2553.35\tabularnewline
CPU & ode12 & semi-automatic & 2839.94 & 1000.24\tabularnewline
CPU & ode23 & automatic & 3122.45 & 1438.54\tabularnewline
CPU & ode23 & semi-automatic & 1170.09 & 803.50\tabularnewline
CPU & ode23s & semi-automatic & 4710.81 & 1877.27\tabularnewline
\midrule
GPU & implicit Euler & conventional & 492.14 & 473.91 \tabularnewline
GPU & implicit Euler & automatic & 541.88 & 527.92 \tabularnewline
GPU & implicit Euler &  semi-automatic & 532.13 & 502.07 \tabularnewline
GPU & ode12 & automatic & 574.26 & 485.00 \tabularnewline
GPU & ode12 & semi-automatic & 600.93 & 481.27 \tabularnewline
GPU & ode23 & automatic & 534.34 & 470.86 \tabularnewline
GPU & ode23 & semi-automatic & 487.55 & 453.94 \tabularnewline
GPU & ode23s & semi-automatic & 958.07 & 468.42 \tabularnewline
\bottomrule
\end{tabular}
\caption{Total runtime of FFT-based homogenization with and without updated reference material for different settings of AutoMat on the CPU and GPU.}
\label{table:total_runtimes}
\end{table*}

\section{Conclusion}
\label{sec:Conclusion}

In this article, we have introduced and studied a universal method for evaluating GSMs. With automatic differentiation, the material law setup is reduced to the implementation of two potentials. This eliminates the inconvenience of hand-computed derivatives and greatly simplifies the material law implementation process.

In a first step, we automatized the conventional implicit Euler approach and were able to reproduce the solution of the elasticity problem up to machine accuracy. However, we also demonstrated that its tangent computation is subject to general accuracy issues. As these can be resolved by an integration of the evolution equation for the state variables with adaptive time step sizes, we detailed how blackbox automatic differentiation of Rosenbrock and Runge-Kutta methods must be modified in the presence of time step size control to obtain derivatives that are as accurate as the primal solution.
Material law evaluations with adaptive time steps improved the solution accuracy of the elasticity problem significantly for large loading steps, especially when stress and stiffnes error measures are used for time step size control. Thus, we have a method at hand to assess the time discretization error disregarding contributions from solving the evolution equation.

To make the method applicable to CT-scale problems, we finally moved the material law evaluation to the GPU. Various kinds of overlap resulted in run times for the stress response that are independent of the chosen integration scheme and are moreover much faster than our conventional implementation on the CPU. 
Especially automatic evaluation strategies are accelerated significantly, which would not be possible without our efficient implementation of automatic differentiation on the GPU.

We conclude that the framework for integrating GSMs into mechanical solvers presented in this article is unmatched in its simultaneous flexibility, accuracy and performance. It is particularly well suited to improve and accelerate matrix-free solvers like FFT-based homogenization. With the resulting user-friendly and fast method, it becomes feasible to investigate the non-linear material behavior, like viscoelasticity and viscoplasticity, of composites on a single workstation.

\begin{acknowledgements}
The authors are deeply indebted to Heiko Andrä for numerous stimulating discussions on the topic of this article. M.~Kabel thanks Quentin Horn for setting up the GPU benchmarking framework of AutoMat during his internship at Fraunhofer ITWM.
\end{acknowledgements}

\bibliographystyle{spbasic}
\bibliography{literature}{}

\end{document}